\documentclass[journal]{IEEEtran}
\usepackage{amsmath,amssymb}
\usepackage{subfigure}
\usepackage{graphicx,graphics,color,psfrag}
\usepackage{cite,balance}
\usepackage{caption}
\allowdisplaybreaks
\usepackage{algorithm}
\usepackage{algorithmic}
\usepackage{accents}
\usepackage{amsthm}
\usepackage{bm}
\usepackage{url}
\usepackage[english]{babel}
\usepackage{multirow}
\usepackage{enumerate}
\usepackage{cases}
\usepackage{stfloats}
\usepackage{dsfont}
\usepackage{color,soul}
\usepackage{amsfonts}
\usepackage{cite,graphicx,amsmath,amssymb}
\usepackage{subfigure}
\usepackage{fancyhdr}
\usepackage{hhline}
\usepackage{graphicx,graphics}
\usepackage{array,color}
\usepackage{mathtools}
\usepackage{amsmath}
\usepackage{hyperref}
\hypersetup{hypertex=true,
colorlinks=true,
linkcolor=blue,
anchorcolor=blue,
citecolor=blue}

\newtheorem{theorem}{\emph{\underline{Theorem}}}

\newtheorem{definition}{\emph{\underline{Definition}}}

\newtheorem{lemma}{\emph{\underline{Lemma}}}

\newtheorem{example}{\bf Example}
\newtheorem{remark}{\bf \emph{\underline{Remark}}}

\def\({\left(}
\def\){\right)}

\def\b0{{\mathbf{0}}}

\newcommand{\nn}{\nonumber}

\begin{document}
\captionsetup[figure]{name={Fig.}}
\title{\huge 
Rotatable Antenna Aided Mixed Near-Field and Far-Field Communications in the Upper Mid-Band: Interference Analysis and Joint Optimization\vspace{-3pt}}

\author{Yunpu~Zhang,
Changsheng~You,~\IEEEmembership{Member,~IEEE}, Hing Cheung So,~\IEEEmembership{Fellow,~IEEE},\\ and Yonina C. Eldar,~\IEEEmembership{Fellow,~IEEE}
\thanks{
Y. Zhang and H. C. So are with the Department of Electrical Engineering, City University of Hong Kong, Hong Kong (e-mail:
yunpu.zhang@my.cityu.edu.hk, hcso@ee.cityu.edu.hk).  
C. You is with the Department of Electronic and Electrical Engineering, Southern University of Science and Technology (SUSTech), Shenzhen
518055, China (e-mail: youcs@sustech.edu.cn).  Yonina C. Eldar is with the Faculty of Math and CS, Weizmann Institute of Science,
Rehovot 76100, Israel (e-mail:
yonina.eldar@weizmann.ac.il).  
 (\emph{Corresponding
 author: Changsheng~You.)}} \vspace{-33pt}}

\maketitle
\begin{abstract}
In this paper, we propose to leverage rotatable antennas (RAs) for improving the communication performance in mixed near-field and far-field communication systems by exploiting a new spatial degree-of-freedom (DoF) offered by antenna rotation to mitigate complex near-field interference and mixed-field interference. Specifically, we investigate an RA-enabled mixed-field downlink communication system with a \emph{modular} architecture, where the base station (BS) is composed of multiple RA subarrays serving multiple near-field users in the presence of legacy far-field users.
We formulate an optimization problem to maximize the sum-rate of the near-field users by 
jointly optimizing the power allocation and rotation angles of all subarrays at the BS. To gain useful insights into the effect of RAs on mixed-field communications, we first analyze a special case where all subarrays share the same rotation angle and obtain closed-form expressions for the \emph{rotation-dependent} normalized near-field interference and the \emph{rotation-dependent} normalized mixed-field interference using the Fresnel integrals. We then analytically reveal that array rotation effectively suppresses both interference types, thereby significantly enhancing mixed-field communication performance. For the general case involving \emph{subarray-wise} rotation, we propose an efficient \emph{double-layer} algorithm to obtain a high-quality solution, where the inner layer optimizes power allocation using the successive convex approximation (SCA), while the outer layer determines the rotation angles of all subarrays via particle swarm optimization (PSO). Finally, numerical results highlight the significant performance gains achieved by RAs over conventional fixed-antenna systems and demonstrate the effectiveness of the developed joint design compared to benchmark schemes.
\end{abstract}
\vspace{-6pt}
\section{Introduction}


To fulfill the ever-growing demands of future sixth-generation (6G) wireless networks for ultra-high data rates, hyper-reliability, and extremely low latency \cite{ 10054381,you2024next}, communication systems must undergo substantial evolution. At the core of this transformation, and critical to shaping 6G and potentially next-generation (NextG) standards, are advancements in multiple-input multiple-output (MIMO) technologies, particularly the deployment of extremely large-scale MIMO (XL-MIMO), along with the strategic utilization of new spectrum bands such as frequency range $3$ (FR3), which spans $7$ GHz to $24$ GHz \cite{10559933}.
These key enabling technologies are anticipated to deliver substantial benefits for 6G, enhancing end-user experiences through improved throughput, extended coverage, and higher energy efficiency. Consequently, they are expected to significantly influence future standardization efforts within bodies like the 3rd Generation Partnership Project (3GPP) \cite{tang2025preliminary}.
These combined thrusts of XL-MIMO and new spectrum utilization are essential for achieving unprecedented spectral efficiency and spatial resolution \cite{10496996,you2023near}. More importantly, this integration fundamentally reshapes the electromagnetic properties of the wireless environment. The synergistic impact of MIMO's evolution towards XL-MIMO, together with the adoption of new spectrum like FR3,  signifies that the transition from conventional far-field communications (planar-wave propagation) to near-field communications (spherical-wave propagation) \cite{9903389,liu2023near}, and ultimately to the more general \emph{mixed near-field and far-field communications} \cite{zhang2023mixed}, is no longer merely a theoretical assumption. Instead, it emerges as an indispensable and practical consideration for air interface design in 6G and NextG systems. For instance, in an XL-MIMO system where a base station (BS) is equipped with $128$ antennas at half-wavelength inter-element spacing operating at $24$ GHz,
the \emph{Rayleigh distance} is approximately $100$ meter (m). Consequently, considering the typical cell sizes in fifth-generation (5G) networks, there are likely users in both the near-field and far-field regions of the BS. This increasingly prevalent mixed-field scenario introduces complex design challenges that demand careful consideration in standardization efforts, including mixed-field interference management \cite{zhang2023mixed} and mixed-field integrated sensing and communications (ISAC) \cite{10812003}. Moreover, understanding and leveraging these mixed-field characteristics will be crucial for realizing novel applications and ensuring their robust operation, facilitating mixed-field user coexistence \cite{10129111} as well as mixed-field simultaneous wireless information and power transfer \cite{zhang2023swipt}. 


More specifically, unlike purely near-field or far-field communications, prior work \cite{zhang2023mixed} has highlighted a key characteristic inherent to mixed-field communications: a near-field user experiences significant interference from beams intended for far-field users, even when positioned at a different spatial angle. This interference, referred to as  \emph{mixed near-field and far-field interference}, arises from the \emph{energy-spread} effect caused by the mismatch between the far-field beamforming vector and the near-field steering vector. This mismatch leads to strong correlations across a wide angular domain, thereby inevitably resulting in complex mixed-field interference challenges. This interference, combined with pure near-field and far-field inter-user interference, significantly complicates the design of mixed-field communication systems. Moreover, existing work on mitigating such interference is still in its infancy, and the related challenges have not been thoroughly investigated.

 Movable antennas (MAs)/fluid antennas \cite{11197972,9264694,9650760,10318061,10243545,10595399,Zhu_Corre,Ma_Corre} have recently emerged as a promising solution for interference management by exploiting the degrees-of-freedom (DoFs) offered by flexible antenna movement. Specifically, the authors in \cite{Zhu_Corre} demonstrated that by properly designing the antenna position, full array gain can be achieved in the desired direction while completely nullifying interference in all undesired directions. Subsequently, this study was further extended in \cite{Ma_Corre} to multiple desired directions, where the beamforming gain across signal directions is maximized while the maximum interference power in undesired directions is constrained. However, the aforementioned works specifically tackled interference mitigation for far-field communications and may not be applicable to mixed-field communication scenarios. In such scenarios, a complex interplay of near-field inter-user interference, far-field inter-user interference, and mixed-field inter-user interference exists, necessitating advanced interference mitigation approaches.

To fully exploit the six-dimensional (6D) spatial DoFs, the concept of 6D movable antennas (6DMAs) has recently been introduced, allowing flexible adjustment of both the three-dimensional (3D) position and 3D rotation of antennas/arrays \cite{shao6d,shao20246d,10883029,10945745}. Although this additional flexibility significantly enhances adaptability to the wireless environment, it also incurs substantial design complexity and implementation cost due to the intricate mechanical constraints associated with antenna movement. To alleviate these challenges, rotatable antennas (RAs) have been proposed as a simplified and cost-effective 6DMA implementation \cite{zheng2025rotatableM,zheng2025rotatable,xie2025thz,dai2025rotatable,zhou2025rotatable,10960698}. RAs can be considered as a special case of 6DMAs, maintaining fixed antenna positions while allowing antenna rotation. Accordingly, RAs greatly reduce hardware complexity yet still offer the capability to reshape radiation patterns and enable flexible coverage. Initial studies have explored the integration of RAs into various wireless systems. For instance, in \cite{zheng2025rotatable}, a new RA model was proposed to enhance the communication performance by independently adjusting each antenna’s boresight direction to achieve a desired array gain pattern. 
In addition, \cite{dai2025rotatable} employs RAs to enable secure data transmission between a BS and a legitimate user in the presence of multiple eavesdroppers, while the authors in \cite{zhou2025rotatable} and \cite{10960698} optimize the orientation of RA array to enhance ISAC performance in single-BS and cooperative BS scenarios, respectively.
 Despite these advances, the potential of RAs in mixed-field communication systems remains largely unexplored, and their ability to tackle the unique interference challenges inherent to mixed-field scenarios has not been thoroughly studied yet. Moreover, the existing literature on RAs mainly focuses on optimizing rotation angles either element-wise or array-wise, which fails to strike a balance between design flexibility and practical implementability. This limitation calls for a more flexible architecture that can better navigate this trade-off.

\begin{figure}[t]
	\centering
\includegraphics[width=0.38\textwidth]{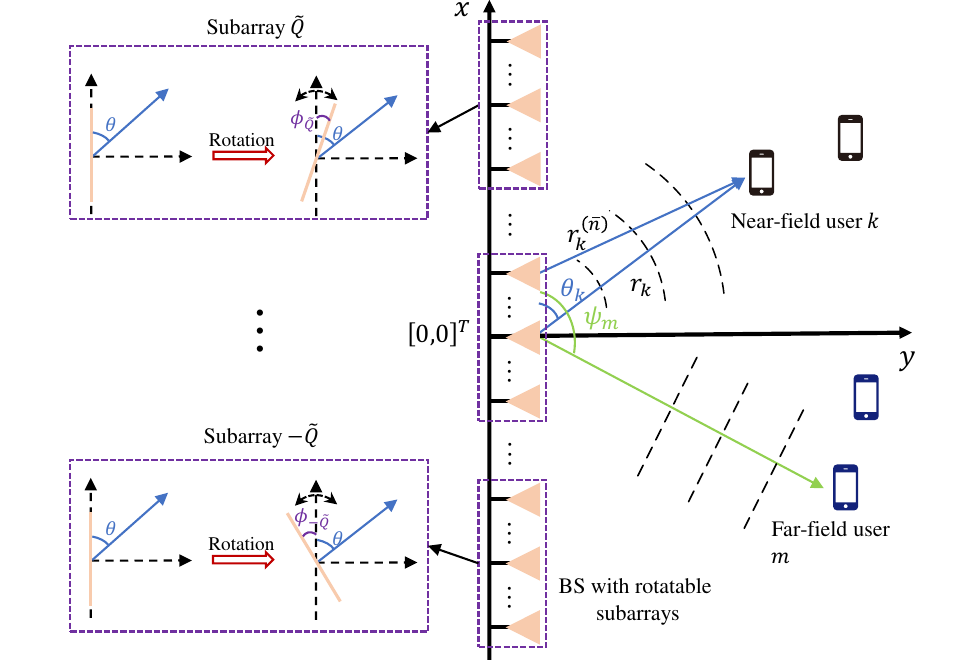}
	\caption{Schematic of modular RA-enabled mixed-field communication system with three modular subarrays.} \label{Fig:systemM}\vspace{-1.7em}
\end{figure}

Motivated by the above, in this paper, we consider a new \emph{modular} RA-enabled mixed near-field and far-field communication system and investigate the performance gains achievable through RAs.
Specifically, we consider a scenario
where a {modular} RA-array, comprising multiple subarrays, is deployed at a BS to serve
multiple near-field users in addition to multiple legacy far-field users supported by preconfigured spatial beams, as illustrated in Fig. \ref{Fig:systemM}. We focus on the joint design of the power allocation for the near-field users and the rotation angles of all RA subarrays at the BS, where the BS is assumed to simultaneously steer
multiple beams towards the near-field users by leveraging the appealing near-field beam-focusing property. The main contributions of the paper are summarized as follows.

\begin{itemize}
    \item First, to the best of our knowledge, we are the first to propose leveraging RAs for improving the communication performance in mixed near-field and far-field communication systems. We formulate an optimization problem to maximize the sum-rate of the near-field users by jointly optimizing the power allocation for the near-field users and the rotation angles of all RA subarrays, subject to constraints on the BS transmit power and the admissible rotation range of each subarray.
    \item Second, to shed light onto the impact of RAs on mixed-field communication performance, we first analyze a special case where all subarrays share the same rotation angle. For this case, we derive closed-form expressions for the \emph{rotation-dependent} normalized near-field interference and the \emph{rotation-dependent} normalized mixed-field interference using the Fresnel integrals. We then analytically demonstrate that array rotation effectively suppresses both near-field and mixed-field interference, thereby significantly enhancing mixed-field communication performance.
    \item Third, for the general case of \emph{subarray-wise} rotation, a non-convex optimization problem with highly coupled variables is obtained. To tackle this,
    we devise an efficient \emph{double-layer} algorithm to obtain a high-quality solution. Specifically, the proposed approach consists of two layers: the inner layer optimizes the power allocation via successive convex approximation (SCA), while the outer layer determines the rotation angles of the subarrays using particle swarm optimization (PSO). 
    \item   Finally, our numerical results highlight the significant benefits of integrating RAs into mixed-field communication systems and validate the effectiveness of the proposed joint design. In particular, we reveal that: 1) RA-enabled systems significantly outperform conventional fixed-antenna systems in mixed-field communication scenarios; 2) subarray-wise rotation offers superior performance compared to uniform array rotation due to greater design flexibility; and 3) the proposed joint design substantially surpasses benchmark schemes that only optimize either power allocation or rotation angles of RA subarrays.
\end{itemize}

 \emph{Notations:} Bold lower-case letters denote column vectors. For a vector $\mathbf{x}$, $[\mathbf{x}]_i$ indicates its $i$-th element. $\mathbb{C}$ represents the set of complex numbers. $\lfloor x \rfloor$ denotes the floor of $x$, i.e., the largest integer less than or equal to $x$.
 The superscripts $(\cdot)^T$ and $(\cdot)^H$ represent the transpose and the conjugate transpose, respectively. $\mathcal{CN}(0,\sigma^2)$ denotes
 the zero-mean circularly symmetric
complex Gaussian (CSCG) distribution with variance $\sigma^2$. The notation $|\cdot|$ denotes the absolute value of a real number or the cardinality of a set. Matrix $\mathbf{I}_{N}$ represents the identity matrix of size $N$. $\mathcal{O}(\cdot)$ denotes the standard big-O notation. 
\vspace{-3pt}
\section{System Model and Problem Formulation}\label{Sec:pro_des}
As illustrated in Fig.~\ref{Fig:systemM}, we consider a mixed-field downlink communication system enabled by \emph{modular} RAs. The system comprises a BS equipped with a uniform linear array (ULA) of $N = 2\tilde{N}+1$ antennas, which serves $K$ single-antenna near-field users indexed by $\mathcal{K}=\{1,2,\ldots,K\}$ in the presence of $M$ legacy single-antenna far-field users indexed by $\mathcal{M}=\{1,2,\ldots,M\}$. The ULA has an aperture of $D=(N-1)d$, where $d=\frac{\lambda}{2}$ represents the inter-element spacing, with $\lambda$ denoting the carrier wavelength.
The distinction between near-field and far-field users is determined by the Rayleigh distance, defined as $Z_{\rm Ray}=\frac{2D^2}{\lambda}$. 
Specifically, near-field users are located at distances smaller than $Z_{\rm Ray}$ from the BS, while far-field users are positioned beyond this threshold.

\vspace{-3pt}
\subsection{Channel Models}
Without loss of generality, we assign the center of the ULA as the origin of the coordinate system, with the $n$-th element, $\forall{n\in \{-\tilde{N},\ldots,\tilde{N}\}}$, located at $[nd,0]^T$. The BS employs a modular RA architecture, partitioning the ULA into $Q = 2\tilde{Q} + 1$ subarrays, each capable of independent one-dimensional (1D) rotation to enhance mixed-field communication performance. 
\subsubsection{Modular RA architecture}
The ULA is divided into $Q$ subarrays, indexed by $\mathcal{Q} = \{-\tilde{Q}, \ldots, \tilde{Q}\}$, with each subarray consisting of $\bar{N}$ antennas. We assume that $\bar{N}$ is an odd number and the antennas within each subarray $q \in \mathcal{Q}$ are indexed by  $\bar{n}\in \mathcal{\bar{N}}=\{-\frac{\bar{N}-1}{2},\ldots,\frac{\bar{N}-1}{2}\}$.\footnote{We assume that each subarray comprises  $\bar{N}=\lfloor \frac{N}{Q}\rfloor$ antennas, with the remaining antennas allocated to the leftmost and rightmost subarrays, respectively.} 
In particular, $Q$ subarrays are connected to a common smart controller \cite{zheng2025rotatable}, wherein
each subarray $q\in \mathcal{Q}$ 
 can be electrically rotated around the $x$-axis by an angle $\phi_q$ to facilitate the \emph{subarray-wise} rotation for improving the mixed-field communication performance \cite{zhou2025rotatable}.\footnote{In this paper, to theoretically evaluate the performance gains of RAs, we assume that the BS employs a ULA with subarray-wise 1D rotation \cite{zhou2025rotatable}. Moreover, the implementation of ULA rotation requires minimal mechanical components in practice, enhancing its practicality for real-world deployment. Notably, the extension to more general cases, including uniform planar array (UPA) and/or three-dimensional (3D) rotation, is straightforward, which will be discussed in detail in Section \ref{Sec:Extension}.} The position of the $\bar{n}$-th antenna, before rotation, is defined relative to the subarray’s center at ${c}_{q}=q\bar{N}d$ along the $x$-axis. After rotation by $\phi_q$,  the antenna’s coordinate become $     \mathbf{w}_{q,\bar{n}} = [c_q+\bar{n}d\cos{\phi_q},\bar{n}d\sin{\phi_q}]^T$, $\forall{q\in \mathcal{Q}}$ and $\bar{n}\in\mathcal{\bar{N}}$. Based on the above,
the channel models for near-field and far-field users accounting for the \emph{subarray-wise} rotation are described as follows.

\subsubsection{Near-field channel model} 
For a typical near-field user, the channel is modeled based on spherical-wave propagation. Consider near-field user $k\in \mathcal{K}$, located at a distance of $r_{k}$ and at a spatial angle of $\theta_k$ from the center of the entire ULA. The coordinates of the $k$-th user are given by $\mathbf{r}_k=[r_k\cos{\theta_k},r_k\sin{\theta_k}]^T$. The
distance from the $\bar{n}$-th antenna in subarray $q$ to user $k$, accounting for the subarray’s rotation, is:
\begin{align}
    &r^{(\bar{n})}_{k} = \|\mathbf{r}_k-\mathbf{w}_{q,\bar{n}}\|\nn\\
    &=
    \sqrt{(r_k\cos{\theta_k}-c_q-\bar{n}d\cos{\phi_q})^2+(r_k\sin{\theta_k}-\bar{n}d\sin{\phi_q})^2}.
\end{align}
Then, the near-field channel steering vector $\mathbf{b}_q(\theta_{k},r_{k},\phi_q)\in\mathbb{C}^{\bar{N}}$ for subarray $q$ is given by
\begin{equation}\label{Eq:NF_steering_sub}
\mathbf{b}_q(\theta_{k},r_{k},\phi_q)=\left[e^{-j \frac{2\pi}{\lambda}(r^{(\bar{n})}_{k}-r_{k})}\right]_{n\in\mathcal{\bar{N}},q\in\mathcal{Q}}.
\vspace{-3pt}
\end{equation}
As such, the overall near-field channel vector, incorporating all subarrays, is given by
\begin{equation}\label{Eq:NF_steering}
	\!\!\mathbf{b}^H(\theta_{k},r_{k},\boldsymbol{\phi})\!\!=\!\!\frac{1}{\sqrt{N}}\!\!\left[\mathbf{b}_{-\tilde{Q}}^H(\theta_{k},r_{k},\phi_{-\tilde{Q}}),\cdots,\mathbf{b}_{\tilde{Q}}^H(\theta_{k},r_{k},\phi_{\tilde{Q}})\right],
\end{equation}
where $\boldsymbol{\phi}=[\phi_{-\tilde{Q}},\cdots,\phi_{\tilde{Q}}]^{T}$ denotes the vector of rotation angles of all subarrays.
In this paper, we adopt a multipath channel model for the near-field users, comprising a
line-of-sight (LoS) path and $L_{{\rm N},k}$ non-line-of-sight (NLoS)
paths induced by scatterers in the environment.
Then, the channel between the BS and near-field user $k$ is modeled as
\begin{equation}
    \!\!\mathbf{h}^H_{{\rm N},k} \!\!=\!\! \sqrt{N}{h}_{{\rm N},k}\mathbf{b}^H(\theta_{k},r_{k},\boldsymbol{\phi})\!+\!\sum^{L_{{\rm N},k}}_{\ell=1}\sqrt{N}{h}_{{\rm N},k,\ell}\mathbf{b}^H(\theta_{k,\ell},r_{k,\ell},\boldsymbol{\phi}),
\end{equation}
where ${h}_{{\rm N},k}$ and ${h}_{{\rm N},k,\ell}$ denote the complex-valued channel gains of the LoS path and the $\ell$-th NLoS path between
the BS and near-field user $k$, respectively.\footnote{For simplicity, we assume an omni-directional antenna gain model. Nonetheless, the proposed scheme can be readily extended to the directional antenna model described in \cite{zheng2025rotatable,zhou2025rotatable}. } $\theta_{k,\ell}$ and $r_{k,\ell}$ are the angle and distance of the $\ell$-th scatterer associated with user $k$, measured relative to the center of the ULA.

\subsubsection{Far-field channel model}
The channel between the BS and far-field user $m\in \mathcal{M}$ follows the planar-wave propagation assumption. Thus, the far-field channel steering vector $\mathbf{a}_q(\psi_m,\phi_q)\in\mathbb{C}^{\bar{N}}$ for subarray $q$ is given by
\begin{equation}\label{Eq:ff_channel}
   \mathbf{a}_q(\psi_m,\phi_q) = \left[e^{-j \frac{2\pi}{\lambda} \bar{n}d\cos{(\psi_m-\phi_q)}}\right]_{\bar{n}\in\mathcal{\bar{N}},q\in\mathcal{Q}},
   \vspace{-3pt}
\end{equation}
where $\psi_m$ denotes the spatial angle of far-field user $m$. Accordingly, the overall far-field channel steering vector is given by
\begin{equation}\label{Eq:FF_steering}
\!\!\!\mathbf{a}^H(\psi_m,\boldsymbol{\phi})\!=\!\frac{1}{\sqrt{N}}\!\left[\mathbf{a}_{-\tilde{Q}}^H(\psi_m,\phi_{-\tilde{Q}}),\cdots,\mathbf{a}_{\tilde{Q}}^H(\psi_m,\phi_{\tilde{Q}})\right].
\vspace{-3pt}
\end{equation}
Consequently, the multipath far-field channel between the BS and the far-field user $m\in\mathcal{M}$ can be expressed as 
\begin{equation}
    \mathbf{h}^H_{{\rm F},m} = \sqrt{N}{h}_{{\rm F},m}\mathbf{a}^H(\psi_m,\phi)+\sum^{L_{{\rm F},m}}_{\ell=1}\sqrt{N}{h}_{{\rm F},m,\ell}\mathbf{a}^H(\psi_{m,\ell},\phi),
\end{equation}
where ${h}_{{\rm F},m}$ and ${h}_{{\rm F},m,\ell}$ denote the complex-valued channel gains of the LoS path and the $\ell$-th NLoS path between
the BS and user $m$, respectively.

\vspace{-12pt}
\subsection{Signal Model}
 Let $s_{{\rm N},k}\in \mathbb{C}$ and $s_{{\rm F},m}\in \mathbb{C}$ represent the signals transmitted by the BS to the $k$-th near-field user with power $P_{{\rm N},k}$ and to the $m$-th far-field user with power $P_{{\rm F},m}$, respectively. 
For \emph{legacy} far-field users, the preconfigured  beamforming vectors are defined as $\mathbf{w}_{{\rm F},m} \in \mathbb{C}^{N}, m \in \mathcal{M}$. For the near-field users, to reduce power consumption and hardware cost in upper mid-band MIMO systems, a hybrid beamforming architecture is employed to serve the $K$ near-field users with $K$ $(K\ll N)$ radio frequency (RF) chains \cite{zhang2023swipt}. 
  The transmitted signal vector is thus given by $\bar{\mathbf{s}}=\mathbf{F}_{\rm A}\mathbf{F}_{\rm D}\mathbf{s}$, where $\bar{\mathbf{s}}=[s_{{\rm N},1},\ldots,s_{{\rm N},K}]^T$, $\mathbf{F}_{\rm D} \in \mathbb{C}^{K\times K}$ is the digital precoder, and $\mathbf{F}_{\rm A}=[\mathbf{v}_{{\rm N},1},\ldots,\mathbf{v}_{{\rm N},K}]\in \mathbb{C}^{N\times K}$ is the analog precoder, with $\mathbf{v}_{{\rm N},k}\in \mathbb{C}^{N}$ representing the analog beamforming vector for the $k$-th near-field user. 
 As such, the received signal at the $k$-th near-field user is given by\footnote{To fully exploit the performance gains brought by RAs, we assume that perfect channel state information (CSI) for both near-field and far-field users is available at the BS \cite{zheng2025rotatable}. In practice, their CSI can be effectively acquired through existing far-field and near-field beam training \cite{zhang2022fast} and/or channel estimation techniques \cite{9598863}.} 
\begin{align}
	y_{{\rm N},k} &= \mathbf{h}^{H}_{{\rm N},k}\sum_{i\in\mathcal{K}} \mathbf{w}_{{\rm N},i}{s}_{{\rm N},i}+ \sum_{m\in\mathcal{M}}\mathbf{h}^{H}_{{\rm N},k}\mathbf{w}_{{\rm F},m}{s}_{{\rm F},m} + n_{k},
\end{align}
where $\mathbf{w}_{{\rm N},i} = \mathbf{F}_{\rm A}\mathbf{f}_{{\rm D},i}$ with $\mathbf{f}_{{\rm D},i}$ being the $i$-th column of $\mathbf{F}_{\rm D}$.
$n_{k} \sim \mathcal{CN}(0,\sigma^2_{k})$ 
represents the additive white Gaussian noise (AWGN) at the near-field user $k$. Accordingly, the achievable rate (bits/s/Hz) of near-field user $k$ is given by $R_{k}=\log_2\left(1+\gamma_{k}\right) $, where
\begin{align}
&\gamma_{k}=\nn\\
&\frac{P_{{\rm N},k}|\mathbf{h}_{{\rm N},k}^H\mathbf{w}_{{\rm N},k}| ^2}{\sum_{i\in\mathcal{K}\setminus\{k\}}P_{{\rm N},i}| \mathbf{h}_{{\rm N},k}^H\mathbf{w}_{{\rm N},i}| ^2+\sum_{m\in\mathcal{M}}P_{{\rm F},m}|\mathbf{h}_{{\rm N},k}^H\mathbf{w}_{{\rm F},m}|^2+\sigma^2_{k}}.
\end{align}

\begin{figure*}[b]
\vspace{-10pt}
   		\hrulefill
        \vspace{-6pt}
	\begin{align}\label{Eq:sum-rate}
&R_{k}\left(\{P_{{\rm N},k}\},\boldsymbol{\phi}\right)=\nn\\
&\log_2\left(1+\frac{P_{{\rm N},k}N|h_{{\rm N},k}|^2}{\sum_{i\in\mathcal{K}\setminus\{k\}}P_{{\rm N},i}N|h_{{\rm N},k}|^2| \mathbf{b}^H(\theta_k,r_k,\boldsymbol{\phi})\mathbf{b}(\theta_i,r_i,\boldsymbol{\phi})|^2+\sum_{m\in\mathcal{M}}P_{{\rm F},m}N|h_{{\rm N},k}|^2|\mathbf{b}^H(\theta_k,r_k,\boldsymbol{\phi})\mathbf{a}(\psi_m,\boldsymbol{\phi})|^2+\sigma^2_{k}}\right).
	\end{align}
\end{figure*}
\vspace{-12pt}
\subsection{Problem Formulation}
To shed light onto the impact of RAs on mixed-field communication systems, we employ a low-complexity hybrid beamforming architecture \cite{zhang2023swipt}. In this architecture, $K$ beams are directed towards the near-field users using near-field beam-focusing, while multiple \emph{preconfigured} maximum
ratio transmission (MRT)-like beams are steered toward the legacy far-field users \cite{10129111} with fixed power allocations. 
The beamforming vectors are defined as $\mathbf{v}_{{\rm N},k} = \frac{\mathbf{h}_{{\rm N},k}}{\|\mathbf{h}_{{\rm N},k}\|}$ for the $k$-th near-field user and $\mathbf{w}_{{\rm F},m} = \frac{\mathbf{h}_{{\rm F},m}}{\|\mathbf{h}_{{\rm F},m}\|}$ for the $m$-th far-field user. The digital beamforming is configured as an identity matrix, i.e., $\mathbf{F}_{\rm D}=\mathbf{I}_{K}$ \cite{zhang2023swipt,dai_LDMA}.\footnote{To enhance the system performance, the weighted minimum mean square error (WMMSE) or zero-forcing (ZF) based digital beamforming can be properly designed based on the analog precoder $\mathbf{F}_{\rm A}$ \cite{zhang2023swipt,dai_LDMA}, with performance evaluations presented in Section \ref{Sec:NR}.}
Consequently, the achievable rate of near-field user $k$ is given by \eqref{Eq:sum-rate} at the bottom of this page.
We consider the system sum-rate maximization problem by jointly optimizing the power allocation for the near-field users $\{P_{{\rm N},k}\}$, and the subarray rotation angles $\boldsymbol{\phi}$. The optimization problem is formulated as follows:
\begin{subequations}
	\begin{align}
		({\bf P1}):~~~\max_{\substack{\{P_{{\rm N},k}\}, \boldsymbol{\phi}} }  &~~	\sum_{k \in \mathcal{K}}~R_{k}\left(\{P_{{\rm N},k}\},\boldsymbol{\phi}\right)
		\\
		\text{s.t.}
		&~~ 	\sum_{k \in \mathcal{K}}P_{{\rm N},k}\le P,\label{P1:pow_cons}\\
        &~~ 	\left|[\boldsymbol{\phi}]_p-[\boldsymbol{\phi}]_q \right|<\frac{\pi}{2} ,~\forall{p,q\in\mathcal{Q}},\label{P1:rot_mutual}\\
        &~~ 	[\boldsymbol{\phi}]_q \in \mathcal{C}_{{\phi_q}},~\forall{q\in\mathcal{Q}},\label{P1:rot_region}
	\end{align}
 \end{subequations}
where \eqref{P1:pow_cons} enforces the BS transmit power constraint, with $P$ being the maximum transmit power. \eqref{P1:rot_mutual} ensures that, for any two subarrays $p,q\in\mathcal{Q}$, their individual rotations $[\boldsymbol{\phi}]_p$ and $[\boldsymbol{\phi}]_q$ form an obtuse angle, thereby preventing mutual signal reflections. \cite{shao6d,zheng2025rotatable}. \eqref{P1:rot_region} restricts the rotation angle of subarray $q$ to its corresponding admissible range $\mathcal{C}_{\phi_q}=[\phi_{q,{\rm min}},\phi_{q,{\rm max}}]$. 

Optimization problem (P1) is non-convex due to the following factors: 1) The objective function is not jointly concave in terms of the optimization variables $\{P_{{\rm N},k}\}$ and $\boldsymbol{\phi}$; 2) these variables are intricately coupled within the objective function; and 3)  the rotation angles $\boldsymbol{\phi}$ introduce complex dependencies in both the near-field and far-field steering vectors.  These characteristics render the optimal solution of (P1) computationally challenging.  To tackle this, we first analyze a special case in Section~\ref{Sec:Special}, where all subarrays employ the same rotation angle, to gain analytical insight into the impact of array rotation on mixed-field communication performance. Subsequently, in Section~\ref{Sec:General}, we address the general case with \emph{subarray-wise} rotation and propose an efficient \emph{double-layer} algorithm to obtain a high-quality suboptimal solution to problem (P1).

\begin{remark}[Extension to other system configurations]
    \emph{In this paper, we analytically investigate the impact of RAs on the performance of mixed-field communication systems, with a particular focus on the suppression of near-field inter-user interference as well as the unique mixed-field inter-user interference. Specifically, we study the maximization of the sum-rate for near-field users in the presence of legacy far-field users. 
    Furthermore, we extend our analysis to explore three additional system configurations:
    \begin{itemize}
\item Maximizing the sum-rate of the far-field users while ensuring the required communication rate for each near-field user. For this configuration, the original optimization problem (P1) needs to incorporate an additional per-user rate constraint. The proposed algorithm remains directly applicable to solve this modified problem.
    \item Maximizing the sum-rate of the near-field users while ensuring that each far-field user meets their communication rate requirement. Similarly, this setup requires problem (P1) to incorporate per-user rate constraints for far-field users, and the proposed algorithm can be directly applied.
    \item Maximizing the sum-rate of both near-field and far-field users. In this case, only the objective function of problem (P1) needs modification to account for the total sum-rate, while the proposed algorithm can still be directly applied.
    \end{itemize}
    }
\end{remark}

\vspace{-15pt}
\section{Analysis of Array Rotation on System Performance}\label{Sec:Special}
In this section, we analytically characterize the effect of RAs on the rate performance of mixed-field communications. To gain useful insights, we begin by considering a special case where all subarrays share the same rotation angle, referred to as \emph{array rotation}. For notational simplicity, we omit the subarray index $q$. Accordingly, we rewrite \eqref{Eq:sum-rate} as \eqref{Eq:sum-rate_SP} at the bottom of this page.
 From \eqref{Eq:sum-rate_SP}, it is clear that the system sum-rate is significantly affected by two types of interference: 
 \emph{near-field inter-user interference}, represented by $| \mathbf{b}^H(\theta_k,r_k,\phi)\mathbf{b}(\theta_i,r_i,\phi)|^2$, and  \emph{mixed-field interference} from the far-field users to near-field users, represented by $|\mathbf{b}^H(\theta_k,r_k,\phi)\mathbf{a}(\psi_m,\phi)|^2$. Moreover, the analytical impact of array rotation on system performance is intricately linked to these two types of interference. Therefore, we proceed by closely examining the influence of array rotation on each interference term in detail.
 To characterize the effects of the rotation angle $\phi\in[\phi_{\rm min},\phi_{\rm max}]$ on these interference terms, we introduce the following definitions.
\begin{figure*}[b]
\vspace{-15pt}
   		\hrulefill
        \vspace{-8pt}
	\begin{align}\label{Eq:sum-rate_SP}
&R_{k}\left(\{P_{{\rm N},k}\},{\phi}\right)=\nn\\
&\log_2\left(1+\frac{P_{{\rm N},k}N|h_{{\rm N},k}|^2}{\sum_{i\in\mathcal{K}\setminus\{k\}}P_{{\rm N},i}N|h_{{\rm N},k}|^2| \mathbf{b}^H(\theta_k,r_k,{\phi})\mathbf{b}(\theta_i,r_i,{\phi})|^2+\sum_{m\in\mathcal{M}}P_{{\rm F},m}N|h_{{\rm N},k}|^2|\mathbf{b}^H(\theta_k,r_k,{\phi})\mathbf{a}(\psi_m,{\phi})|^2+\sigma^2_{k}}\right).
	\end{align}
\end{figure*}
\begin{definition}
    \emph{ The \emph{rotation-dependent} normalized near-field inter-user interference between any near-field users $k$ and $i$, $\forall{k,i\in\mathcal{K}}$ is defined as:
    \begin{equation}
\rho_{\rm NN}(\phi,\theta_k,r_k,\theta_i,r_i)\triangleq|\mathbf{b}^H(\theta_k,r_k,\phi)\mathbf{b}(\theta_i,r_i,\phi)|.
 \end{equation}}
\end{definition}
\begin{definition}
    \emph{ The \emph{rotation-dependent} normalized mixed-field inter-user interference between the far-field user $m\in\mathcal{M}$ and near-field user $k\in \mathcal{K}$ is defined as:
    \begin{equation}
\rho_{\rm NF}(\phi,\theta_k,r_k,\psi_m)\triangleq|\mathbf{b}^H(\theta_k,r_k,\phi)\mathbf{a}(\psi_m,\phi)|.
 \end{equation}}
\end{definition}
Moreover, it follows directly from \eqref{Eq:sum-rate_SP} that the sum-rate is a monotonically decreasing function of both $\rho_{\rm NN}(\cdot)$ and $\rho_{\rm NF}(\cdot)$.
Mathematically, these interference terms are explicitly given in \eqref{Eq:NN_inter} and \eqref{Eq:NF_inter}, as shown at the top of the next page. These expressions, however, are complex and difficult to analyze directly. To address this, we obtain closed-form expressions by using the Fresnel integrals, as presented below.
\begin{figure*}[t]
	\begin{equation}\label{Eq:NN_inter}
\rho_{\rm NN}(\phi,\theta_k,r_k,\theta_i,r_i)
=\frac{1}{N}\left|\sum_{n=-\tilde{N}}^{\tilde{N}}e^{j\frac{2\pi}{\lambda}\left[ n^2\left(\frac{d^2\sin^2{(\phi-\theta_k)}}{2r_k}-\frac{d^2\sin^2{(\phi-\theta_i)}}{2r_i}\right)+n(d\cos{(\phi-\theta_i)}- d\cos{(\phi-\theta_k)}) \right]}\right|.
 \vspace{-3pt}
	\end{equation}
    		\hrulefill
            \vspace{-10pt}
\end{figure*}
\begin{figure*}[t]
	\begin{equation}\label{Eq:NF_inter}
\rho_{\rm NF}(\phi,\theta_k,r_k,\psi_m)
=\frac{1}{N}\left|\sum_{n=-\tilde{N}}^{\tilde{N}}e^{j\frac{2\pi}{\lambda}\left[ n^2\frac{d^2\sin^2{(\phi-\theta_k)}}{2r}+n(d\cos{(\psi_m-\phi)}- d\cos{(\phi-\theta_k))} \right]}\right|. \vspace{-5pt}
	\end{equation}
    		\hrulefill
            \vspace{-20pt}
\end{figure*}

\vspace{-8pt}
\subsection{Closed-form Expressions for Both Interference Terms}
A closed-form approximation for the \emph{rotation-dependent} normalized near-field inter-user interference in \eqref{Eq:NN_inter} is provided in Theorem \ref{The:NN_inter}.
\begin{theorem}\label{The:NN_inter}
    \emph{The \emph{rotation-dependent} normalized near-field inter-user interference $\rho_{\rm NN}(\phi,\theta_k,r_k,\theta_i,r_i)$ in \eqref{Eq:NN_inter} can be approximated as:
    \begin{equation}\label{Eq:NN_inter_approx}
        \rho_{\rm NN}(\phi,\theta_k,r_k,\theta_i,r_i) \approx G(\beta_1,\beta_2) = \left|\frac{\widehat{C}(\beta_1,\beta_2)+j\widehat{S}(\beta_1,\beta_2)}{2\beta_2}\right|,
    \end{equation}
    where 
     \begin{align}\label{Eq:NN_inter_beta}
\beta_1&=\frac{\left(\cos{(\phi-\theta_k)}-\cos{(\phi-\theta_i)}\right)}{\sqrt{d\left|\frac{\sin^2{(\phi-\theta_k)}}{r_k}-\frac{\sin^2{(\phi-\theta_i)}}{r_i}\right|}}, \nn\\
\beta_2&=\frac{N}{2} \sqrt{d\left|\frac{\sin^2{(\phi-\theta_k)}}{r_k}-\frac{\sin^2{(\phi-\theta_i)}}{r_i}\right|}.
    \end{align}
     Moreover, $\widehat{C}(\beta_1,\beta_2) = {C}(\beta_1+\beta_2)-{C}(\beta_1-\beta_2)$ and $\widehat{S}(\beta_1,\beta_2) = {S}(\beta_1+\beta_2)-{S}(\beta_1-\beta_2)$, where $C(\cdot)$ and $S(\cdot)$ are the Fresnel integrals, defined as $C(x)=\int^{x}_{0}\cos(\frac{\pi}{2}t^2)\text{d}t$ and $S(x)=\int^{x}_{0}\sin(\frac{\pi}{2}t^2)\text{d}t$, respectively.
    }
\end{theorem}
\begin{proof}
     Please refer to Appendix \ref{App1}.
\end{proof}

Next, we provide a closed-form approximation for the \emph{rotation-dependent} normalized mixed-field inter-user interference. 

\begin{theorem}\label{The:NF_inter}
    \emph{The \emph{rotation-dependent} normalized mixed-field inter-user 
 interference $\rho_{\rm NF}(\phi,\theta_k,r_k,\psi_m)$ in \eqref{Eq:NF_inter} can be approximated as
    \begin{equation}\label{Eq:NF_inter_appro}
        \rho_{\rm NF}(\phi,\theta_k,r_k,\psi_m) \approx G(\beta_1,\beta_2) = \left|\frac{\widehat{C}(\beta_1,\beta_2)+j\widehat{S}(\beta_1,\beta_2)}{2\beta_2}\right|,
    \end{equation}
    where 
     \begin{align}\label{Eq:NF_inter_beta}
        \beta_1 &=(\cos{(\phi-\theta_k)}-\cos{(\psi_m-\phi)})\sqrt{\frac{r_k}{d\sin^2{(\phi-\theta_k)}}},\nn\\
        \beta_2 &=\frac{N}{2} \sqrt{\frac{d\sin^2{(\phi-\theta_k)}}{r_k}}.
    \end{align}
    }
\end{theorem}
\begin{proof}
     The proof is similar to that of Theorem \ref{The:NN_inter}, and hence is omitted for brevity.
\end{proof}

\begin{remark}\label{Re:Properites}
    \emph{It is worth highlighting that Theorem \ref{The:NN_inter} generalizes Theorem \ref{The:NF_inter}, reducing to it as  $r_k$ (or $r_i$) approaches infinity, i.e., $r_k ({\rm or}~r_i) \rightarrow \infty$ . Moreover, since both closed-form expressions accounting for the near-field and mixed-field interference are derived capitalizing on the Fresnel integrals, they share the same properties of function $G(\cdot)$, which facilitate the subsequent interference analysis, as detailed below.
     \begin{itemize}
        \item First, when either parameter $\beta_1$ or $\beta_2$ equals zero, the function $G(\cdot)$ attains its maximum value  \cite{zhang2023mixed,polk1956optical}; if both are zero, the maximum value is one.
        \item Second, the function is symmetric with respect to $\beta_1$ and generally decreases as $|\beta_1|$ increases on both sides, while it generally decreases as $\beta_2$ increases \cite{dai_LDMA}. 
        \item Finally, when $\phi=0$ (fixed-antenna system), both expressions of near-field and mixed-field interference, i.e., $\rho_{\rm NN}(\cdot)$ and $\rho_{\rm NF}(\cdot)$, simplify to their fixed forms:
\begin{itemize}
    \item[1)] $\rho_{\rm NN}(0,\theta_k,r_k,\theta_i,r_i)= G(\bar{\beta}_1,\bar{\beta}_2)$ \cite{zhang2023swipt}, where 
    \begin{align}
          \bar{\beta}_1 &=\frac{\left(\cos{\theta_k}-\cos{\theta_i}\right)}{\sqrt{d\left|\frac{\sin^2{\theta_k}}{r_k}-\frac{\sin^2{\theta_i}}{r_i}\right|}},\nn\\
          \bar{\beta}_2&=\frac{N}{2}  \sqrt{d\left|\frac{\sin^2{\theta_k}}{r_k}-\frac{\sin^2{\theta_i}}{r_i}\right|}.
    \end{align}
    \item[2)] $\rho_{\rm NF}(0,\theta_k,r_k,\psi_m)= G(\tilde{\beta}_1,\tilde{\beta}_2)$ \cite{zhang2023mixed}, where 
    \begin{align}
        \tilde{\beta}_1 &=(\cos{\theta_k}-\cos{\psi_m})\sqrt{\frac{r_k}{d\sin^2{\theta}}},\nn\\
\tilde{\beta}_2&=\frac{N}{2} \sqrt{\frac{d\sin^2{\theta}}{r_k}}.
    \end{align}
\end{itemize}
    \end{itemize}}
\end{remark}

Theorems \ref{The:NN_inter} and \ref{The:NF_inter} reveal an important insight: compared to conventional fixed-antenna systems, the rotation angle $\phi$ introduces an additional DoF for controlling the values of $\beta_1$ and $\beta_2$, and consequently, the interference levels measured by $G(\cdot)$. This manifests that the added flexibility enables the suppression of both near-field and mixed-field interference through strategic adjustment of the rotation angle $\phi$. In the subsequent subsections, we analytically demonstrate how the rotation angle $\phi$ affects the near-field and mixed-field interference, respectively, and hence the overall system performance. 

\subsection{Near-field Inter-user Interference}\label{Sec:NNInter}
We now examine how the array rotation angle $\phi$ impacts the near-field interference.

\begin{lemma}\label{Le:NN_equal}
    \emph{For two near-field users located at the same angle, i.e., $\theta_k=\theta_i=\theta$, but different distances, i.e., $r_k\neq r_i$, the \emph{rotation-dependent} normalized near-field interference in \eqref{Eq:NN_inter_approx} depends solely on $\beta_2$ in \eqref{Eq:NN_inter_beta}, such that $ \rho_{\rm NN} \approx G(0,\beta_2) $.
    The effect of array rotation on this interference can be divided into the following two cases:
    \begin{itemize}
        \item If $\theta =  \frac{\pi}{2}$, array rotation does not reduce near-field interference, as it cannot increase $\beta_2$. The optimal rotation angle is:
        \begin{equation}
            \phi^* = 0.
        \end{equation}
        \item If $\theta \neq \frac{\pi}{2}$, array rotation 
        reduces the near-field interference by maximizing $\beta_2$. The optimal rotation angle is:
        \begin{equation}
            \phi^* = \begin{cases}
            \min(\theta-\frac{\pi}{2},\phi_{\rm max}), & {\rm if}~ \theta>\frac{\pi}{2}, \\
                 \max(\theta-\frac{\pi}{2},\phi_{\rm min}), & {\rm if}~ \theta<\frac{\pi}{2}.
            \end{cases}
        \end{equation}
    \end{itemize}}
\end{lemma}
\begin{proof}
Please refer to Appendix \ref{App11}.
\end{proof}

\begin{lemma}\label{Le:NN_equal_dist}
    \emph{
    For two near-field users located at the same distance, i.e., ${r_k}={r_i}$, but different angles, i.e., $\theta_k\neq\theta_i$, there always exists a rotation angle $\phi$ such that 
        \begin{equation}
        \beta_1 > \bar{\beta}_1,
        \beta_2 > \bar{\beta}_2
  \Longrightarrow G(\beta_1,\beta_2) < G(\bar{\beta}_1,\bar{\beta}_2).
    \end{equation}
    }
\end{lemma}

\begin{proof}
    Please refer to Appendix \ref{App2}.
\end{proof}
   \begin{figure*}[t!]
\begin{minipage}{.245\textwidth}
	\centering
\includegraphics[width=1\columnwidth]{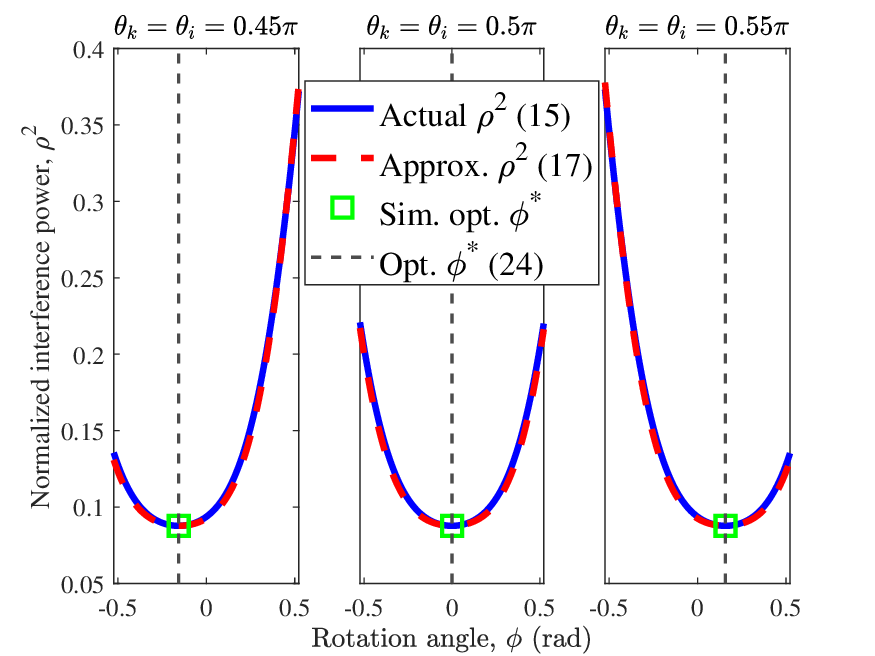}
	\caption{{Normalized near-field interference power versus rotation angle.}\label{Fig:NNverify}} 
    \end{minipage}	
    \hfill
    \begin{minipage}{.245\textwidth}
	\centering
\includegraphics[width=1\columnwidth]{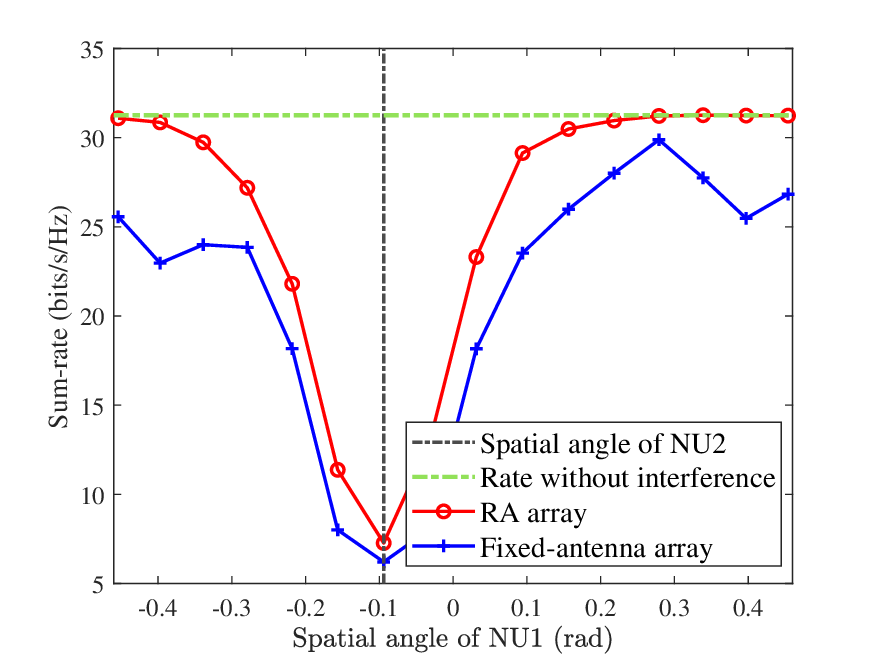}
	\caption{{Sum-rate versus spatial angle of NU1 with $r_{{\rm N},1} = 0.05Z_{\rm Ray}$ and $r_{{\rm N},2} = 0.08Z_{\rm Ray}$.}\label{Fig:NNangle}} 
    \end{minipage}	
    \hfill
        \begin{minipage}{.245\textwidth}
\centering
\includegraphics[width=1\columnwidth]{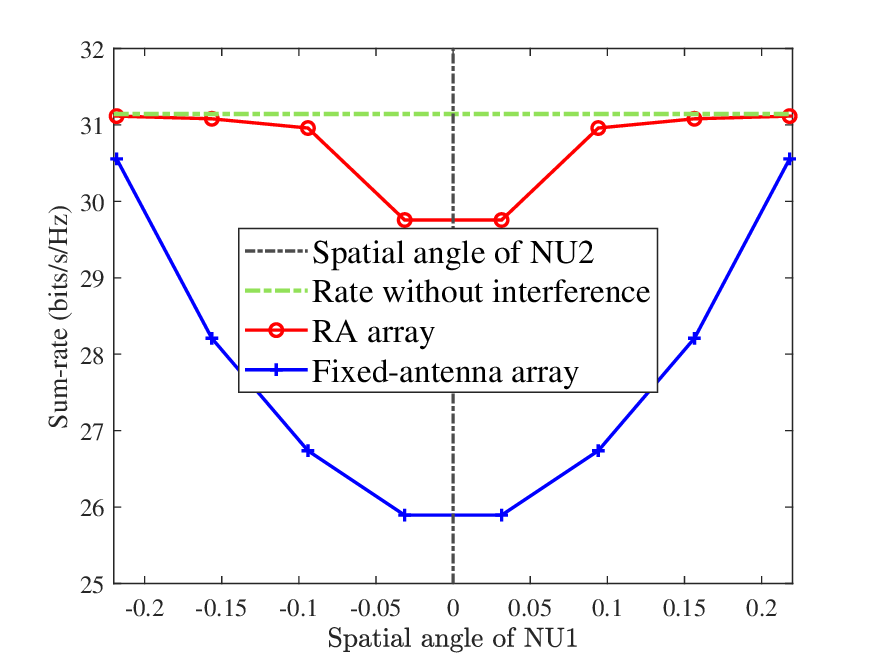}
	\caption{{Sum-rate versus spatial angle of NU1 with same distance  $r_{{\rm N},1} =r_{{\rm N},2}=0.05Z_{\rm Ray}$.}\label{Fig:NNdist1}} 
    \end{minipage}	
    \begin{minipage}{.245\textwidth}
\centering
\includegraphics[width=1\columnwidth]{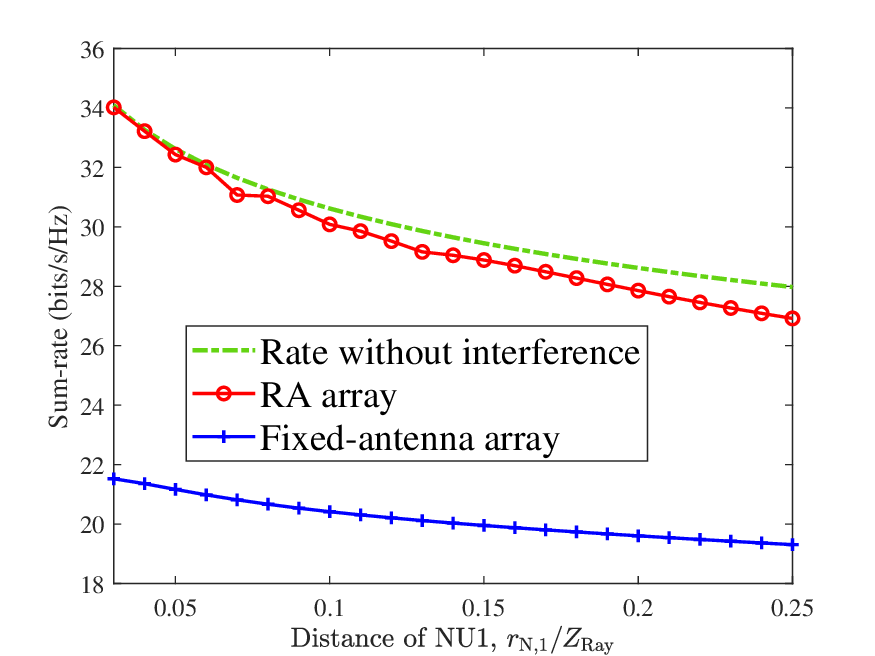}
	\caption{{Sum-rate versus distance of NU1 with $\theta_{{\rm N},1} = 0.55\pi$, $\theta_{{\rm N},2} =0.6\pi$, and $r_{{\rm N},2}=0.05Z_{\rm Ray}$.}\label{Fig:NNdist2}} 
    \end{minipage}	
    \vspace{-18pt}
\end{figure*}
\begin{remark}[Effect of RAs on near-field interference]\label{Re:effectRANN}
    \emph{Lemmas \ref{Le:NN_equal} and \ref{Le:NN_equal_dist} reveal several key insights into how RAs mitigate near-field interference:
   \begin{itemize}
       \item The extra DoF introduced by RAs enhances mixed-field communication performance by suppressing the near-field interference in nearly all scenarios. However, an exception occurs when both users are positioned at the boresight direction (i.e., $\theta_k=\theta_i=\frac{\pi}{2}$). In this case, array rotation yields no performance gain, regardless of the users' distances from the array, as it equally affects both near-field steering vectors, rendering it ineffective at reducing normalized near-field interference.
    \item  When both users share the same spatial angle but not at the boresight direction (i.e., $\theta_k=\theta_i\neq\frac{\pi}{2}$), the performance gains from array rotation depend on the spatial angle. These gains diminish as the angle approaches the boresight direction from either side, indicating that rotation becomes less effective at differentiating steering vectors as they increasingly align with the boresight.
    \item When both users are at the same distance from the array but have different spatial angles (i.e., $\theta_k\neq\theta_i$), the performance improvement due to array rotation becomes more pronounced. The distinct angles allow rotation to effectively exploit differences in the steering vectors, resulting in enhanced interference suppression. This effect will be revealed numerically in the subsequent example.
   \end{itemize}
    }
\end{remark}
\begin{example}[How does the rotation angle affect the near-field interference?]
    \emph{
    To demonstrate the benefits of RAs in mitigating the near-field inter-user interference, we first validate the accuracy of the proposed approximation in Theorem \ref{The:NN_inter} and the results derived in Lemma \ref{Le:NN_equal}. This corroboration is
    presented in Fig. \ref{Fig:NNverify}, where the Fresnel-based approximation is shown to closely match the exact normalized interference given in \eqref{Eq:NN_inter}. Moreover,
    the derived rotation angles align closely with the numerical values across various spatial-angle settings. Next, 
    we investigate the effects of the rotation angle on the achievable sum-rate with respect to the spatial angle and distance of the near-field user. This analysis is detailed in Figs. \ref{Fig:NNangle} -- \ref{Fig:NNdist2}, where the system parameters are set to $N=129$, $f=24$ GHz, and $K=2$ (i.e., near-field users NU1 and NU2).
   We see that RA-enabled systems consistently outperform their fixed-antenna counterparts across all tested angles and distances. We explore the effect of rotation on the achievable sum-rate from both angular and distance perspectives as follows:
    \begin{itemize}
        \item \emph{Effect of angle given different distances:} Fig. \ref{Fig:NNangle} shows that when two near-field users are positioned at the same angle (but not at the boresight direction), the performance gain from RAs is minimal. This is due to the peak inter-user interference caused by significant signal overlap in the angular domain. In such cases, array rotation can only adjust $\beta_2$ to partially mitigate this interference, resulting in limited improvement. However, as the angular separation increases, substantial performance enhancement is observed. With a large angular difference, the RA-enabled system rate approaches the interference-free bound, meaning that interference is nearly eliminated. This aligns with the insights in Remark \ref{Re:effectRANN}, which suggests that RAs are particularly effective when users are spatially separated in the angular domain, allowing rotation to steer the interference away from the desired signal.
        \item \emph{Effect of angle given the same distance:} Fig. \ref{Fig:NNdist1} illustrates that when two near-field users are located at the same distance, the performance gain from RAs depends heavily on their angular separation. Notably, significant gains are observed when the angular separation is small, where interference is inherently strong, and RAs can effectively reduce it by adjusting the array’s orientation. As the angular separation increases, both RA-enabled and fixed-antenna systems benefit from reduced interference due to the increased orthogonality of the near-field and far-field steering vectors.
    \item \emph{Effect of distance given different angles:} Fig. \ref{Fig:NNdist2} reveals an interesting result that when the two users have different spatial angles, RA-enabled systems achieve near interference-free performance at smaller near-field distances. However, as the distance of NU1 increases, this advantage slightly diminishes. This is because at small distances, the near-field spatial resolution is highly pronounced, indicating that even a slight distance difference between users significantly reduces inter-user interference. As NU1’s distance increases, the near-field effect diminishes, and the interference transitions into a more complex, mixed-field-like scenario. This shift reduces the effectiveness of rotation, leading to a performance drop relative to the interference-free rate bound. This phenomenon will be investigated in detail in Section \ref{Sec:mixed-field}, where the transition from near-field to mixed-field interference is analytically modeled.
    \end{itemize}}
\end{example}
\vspace{-15pt}
\subsection{Mixed-field Inter-user Interference}\label{Sec:mixed-field}  
In this subsection, we characterize the effect
of rotation angle $\phi$ on mitigating the mixed-field inter-user
interference. Since the interference characteristic between a far-field user $m$ and a near-field user $k$ applies universally to any near-field and far-field user pair, we omit these indices for notational convenience.

\begin{lemma}\label{Le:equalAngle}
    \emph{When the far-field user and the near-field user are at the same angle, i.e., $\psi=\theta$, the \emph{rotation-dependent} normalized mixed-field interference in \eqref{Eq:NF_inter_appro} depends only on $\beta_2$ in \eqref{Eq:NF_inter_beta}, such that $ \rho_{\rm NF} \approx G(0,\beta_2) $.
    The effect of rotation on this interference can be characterized as follows:
    \begin{itemize}
        \item If $\psi = \frac{\pi}{2}$, array rotation yields no reduction in interference, as $\beta_2$ cannot be increased. The optimal rotation angle is:
        \begin{equation}\label{Eq:boresight}
            \phi^* = 0.
        \end{equation}
        \item If $\psi\neq \frac{\pi}{2}$, array rotation reduces interference by maximizing $\beta_2$. The optimal rotation angle is:
        \begin{equation}\label{Eq:Nboresight}
            \phi^* = \begin{cases}
            \min(\psi-\frac{\pi}{2},\phi_{\rm max}), & {\rm if}~ \psi>\frac{\pi}{2}, \\
                 \max(\psi-\frac{\pi}{2},\phi_{\rm min}), & {\rm if}~ \psi<\frac{\pi}{2}.
            \end{cases}
        \end{equation}
    \end{itemize}
    }
\end{lemma}
\begin{proof}
     The proof is similar to that of Lemma \ref{Le:NN_equal}, and hence is omitted for brevity.
\end{proof}
\begin{lemma}\label{Le:reduce_corre}
    \emph{In the general case of different spatial angles, there always exists a rotation angle $\phi$, such that $\beta_1 > \tilde{\beta}_1$ and $\beta_2 > \tilde{\beta}_2$, implying:
    \begin{equation}
        \beta_1 > \tilde{\beta}_1,
        \beta_2 > \tilde{\beta}_2 \Longrightarrow G(\beta_1,\beta_2) < G(\tilde{\beta}_1,\tilde{\beta}_2).
    \end{equation}
    }
\end{lemma}
\begin{proof}
     The proof is similar to that of Lemma \ref{Le:NN_equal_dist}, and hence is omitted for brevity.
\end{proof}

\begin{remark}[Effect of RAs on mixed-field interference.]\label{Re:effectRA}
    \emph{Lemmas \ref{Le:equalAngle} and \ref{Le:reduce_corre} reveal that:
   \begin{itemize}
       \item The additional DoF introduced by RAs enhances mixed-field communication performance by suppressing the mixed-field interference in nearly all scenarios, except when both users are at the boresight direction (i.e., $\theta=\psi=\frac{\pi}{2}$).  At this boresight alignment, rotation equally affects both near-field and far-field steering vectors, making it ineffective at reducing normalized interference.
    \item  When the near-field and far-field users share the same spatial angle but are not at the boresight direction (i.e., $\theta=\psi\neq\frac{\pi}{2}$), performance gains from array rotation depend on the spatial angle, which diminish when the angle approaches the boresight direction from either side.
    \item When the spatial angles of the near-field and far-field users differ  (i.e., $\theta \neq \psi$), the performance improvement is more pronounced, as demonstrated numerically in the following example.
   \end{itemize}
    }
\end{remark}

    \begin{figure*}[t!]
\begin{minipage}{.325\textwidth}
	\centering
\includegraphics[width=1\columnwidth]{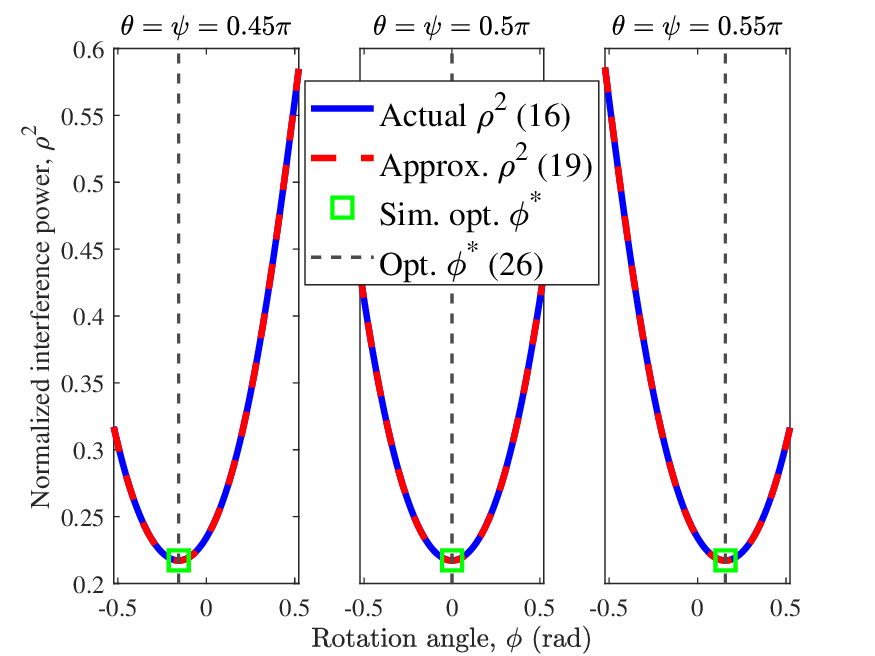}
	\caption{{Normalized mixed-field interference power versus rotation angle.}\label{Fig:verify}} 
    \end{minipage}	
    \hfill
    \begin{minipage}{.325\textwidth}
	\centering
\includegraphics[width=1\columnwidth]{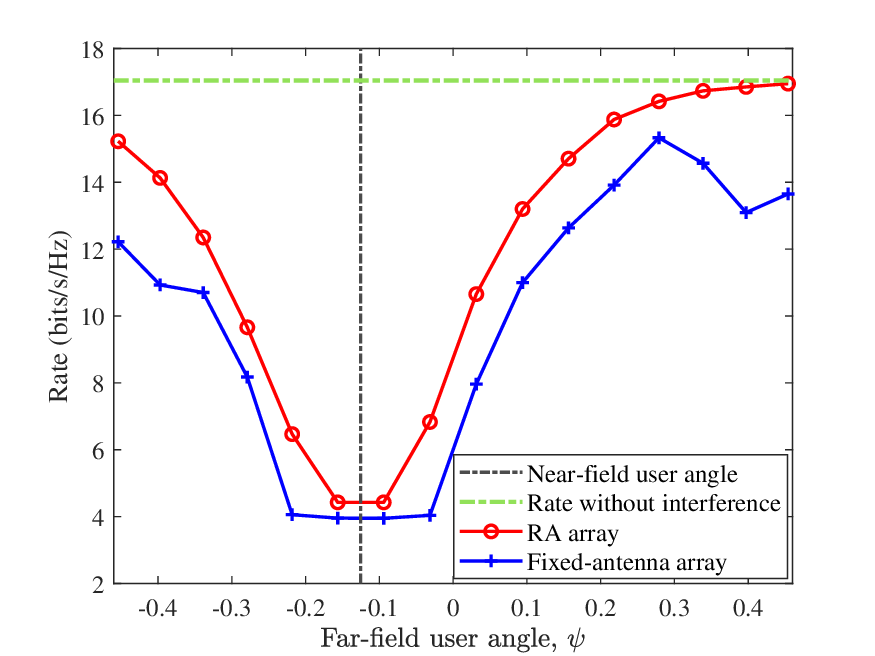}
	\caption{{Rate versus far-field user angle with $r_{\rm N} = 0.05Z_{\rm Ray}$.}\label{Fig:angle}} 
    \end{minipage}	
    \hfill
        \begin{minipage}{.325\textwidth}
\centering
\includegraphics[width=1\columnwidth]{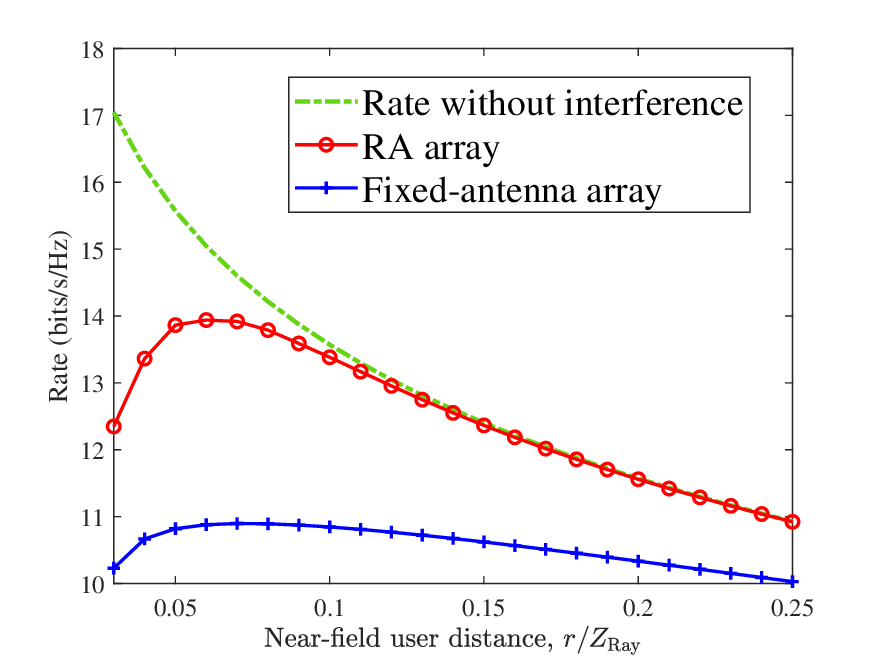}
	\caption{{Rate versus near-field user distance with  $\theta = 0.55\pi$ and $\psi=0.6\pi$.}\label{Fig:dist}} 
    \end{minipage}	
    \vspace{-18pt}
\end{figure*}

\begin{example}[How does the rotation angle affect the mixed-field interference?]
    \emph{
    To highlight the benefits of RAs in enhancing mixed-field communication performance, we first validate the accuracy of the proposed approximation in Theorem \ref{The:NF_inter} and the results derived in Lemma \ref{Le:equalAngle}, which is illustrated in Fig. \ref{Fig:verify}. We observe that the Fresnel-based approximation coincides with the exact normalized interference given in \eqref{Eq:NF_inter}. Furthermore,
    the analytically derived rotation angles align well with numerical results across various spatial-angle settings. Next, 
    we study the effects of the rotation angle on the achievable rate from two perspectives: the spatial angle of the far-field user and the distance of the near-field user. The results are detailed in Figs. \ref{Fig:angle} and \ref{Fig:dist}, with system parameters $N=129$ and $f=24$ GHz.
    It is demonstrated that RA-enabled systems outperform their fixed-antenna counterparts across all tested angles and distances. Specifically, Fig. \ref{Fig:angle} shows that the performance improvement is marginal when the near-field user's angle is close to that of the far-field user. However, as the angular separation increases, the gain becomes more pronounced. When the angular difference is large, the RA-enabled system rate approaches the interference-free bound, consistent with  Remark \ref{Re:effectRA}.
    Furthermore, Fig. \ref{Fig:dist} illustrates that when the spatial angles of the near-field and far-field users differ, RAs provide substantial gains at smaller near-field distances. As the near-field distance reaches approximately $r_{\rm N} = 0.15Z_{\rm Ray}$, the achievable rate reaches the interference-free bound, reflecting the reduced interference brought by rotation.
}
\end{example}

\vspace{-10pt}
\subsection{Extension to General Scenarios}\label{Sec:Extension}
The analytical results presented in Sections \ref{Sec:NNInter} and \ref{Sec:mixed-field} are derived for a ULA with array rotation. However, it is important to note that the analytical framework can be readily extended to more general scenarios, such as a ULA with 3D rotation \cite{xie2025thz} or a UPA element-wise 3D rotation \cite{zheng2025rotatable}. These extensions are outlined below.
\subsubsection{ULA with 3D rotation} In this scenario, the orientation of the ULA is controlled by three rotation angles, denoted as $[\alpha,\beta,\gamma]^T$, where $\alpha$, $\beta$, and $\gamma$ correspond to rotations around the $x$-, $y$-, and $z$-axes, respectively. By varying one rotation angle while keeping the others fixed, the impact on mixed-field interference can be theoretically characterized. In particular, closed-form expressions based on the Fresnel integrals can be derived to quantify how each rotation angle individually affects mixed-field communication performance.
\subsubsection{UPA with element-wise 3D rotation} For a UPA, each antenna element can be independently oriented in 3D space, with its orientation defined by a pair of deflection angles: the zenith and azimuth angles. Similarly, the influence of these orientation parameters on mixed-field communication performance can similarly be analyzed in closed form by using the Fresnel integrals. For example, by fixing the azimuth angle and varying the zenith angle of a specific element, or vice versa, one can assess how individual angle affects the overall system behavior and performance.

More importantly, it is worth noting that, while the assumption of \emph{uniform subarray rotation} facilitates both analytical tractability and practical implementation, it compromises design flexibility
 compared to the \emph{subarray-wise} rotation. To fully exploit the performance gains of RAs, we in the next delve into developing the solution for the general case with subarray-wise rotation.

 \vspace{-10pt}
\section{Proposed Solution to Problem (P1)}\label{Sec:General}
In this section, we present an efficient \emph{double-layer} algorithm to solve problem (P1) for the general case with \emph{subarray-wise} rotation.
Specifically, we reformulate problem (P1) equivalently by
decomposing it into two subproblems: an inner problem for optimizing the power allocation $\{P_{{\rm N},k}\}$ for fixed rotation angles $\boldsymbol{\phi}$, and an outer problem for optimizing the rotation angles $\boldsymbol{\phi}$. These subproblems are detailed below.
\vspace{-9pt}
\subsection{Inner Problem: Power Allocation Optimization}
For any feasible rotation angles $\boldsymbol{\phi}$,
the inner problem aims to optimize the power allocation $\{P_{{\rm N},k}\}$ by minimizing the negative sum-rate, formulated as:
\begin{subequations}
	\begin{align}
		({\bf P2}):~~~\min_{\substack{\{P_{{\rm N},k}\}} }  ~~	-\sum_{k \in \mathcal{K}}~R_{k}(\{P_{{\rm N},k}\}),
		~~~~\text{s.t.}~~~
		\eqref{P1:pow_cons}.\nn
	\end{align}
 \end{subequations}
The primary challenge in solving problem (P2) lies in the non-convexity of the objective function. To address this, we rewrite the objective function as a difference of two convex functions: $-\sum_{k\in\mathcal{K}}R_k(\{P_{{\rm N},k}\}) = A_1-B_1$, where $A_1$ and $B_1$ are given by \begin{align}\label{Eq:expression1}
    &A_1 =- \sum_{k\in\mathcal{K}}\log_2\bigg(\sum_{i\in\mathcal{K}}P_{{\rm N},i} N|h_{{\rm N},k}|^2\rho^2_{\rm NN}(\boldsymbol{\phi},\theta_k,r_k,\theta_i,r_i)
    \nn\\&+\sum_{m\in\mathcal{M}}P_{{\rm F},m}N|h_{{\rm N},k}|^2\rho^2_{\rm NF}(\boldsymbol{\phi},\theta_k,r_k,\psi_m)+\sigma^2_{k}\bigg),\nn\\
   &B_1 =- \sum_{k\in\mathcal{K}}\log_2\bigg(\sum_{i\in\mathcal{K}\setminus\{k\}}P_{{\rm N},i}N|h_{{\rm N},k}|^2\rho^2_{\rm NN}(\boldsymbol{\phi},\theta_k,r_k,\theta_i,r_i)
    \nn\\&+\sum_{m\in\mathcal{M}}P_{{\rm F},m}N|h_{{\rm N},k}|^2\rho^2_{\rm NF}(\boldsymbol{\phi},\theta_k,r_k,\psi_m)+\sigma^2_{k}\bigg).
\end{align}
 Both functions, $A_1$ and $B_1$, are jointly convex in terms of $\{P_{{\rm N},k}\}$. To proceed,  we employ the SCA method, iteratively constructing a convex upper bound for the objective function. For effective implementation of SCA, we develop a global underestimator for $B_1$, using the superscript $(t)$ to indicate the iteration index of the optimization variables. Specifically, for any feasible point $\mathbf{p}^{(t)}\triangleq \left\{P^{(t)}_{{\rm N},k}\right\}^{K}_{k=1}$, we derive a lower bound for $B_1$ through its first-order Taylor approximation, expressed as 
 \begin{equation}
     B_1(\mathbf{p})\ge B_1\left(\mathbf{p}^{(t)}\right)+\sum_{k\in\mathcal{K}}\nabla_{P_{{\rm N},k}}B_1\left(\mathbf{p}^{(t)}\right)\left(P_{{\rm N},k}-P_{{\rm N},k}^{(t)}\right),
     \end{equation}
where the gradient of $B_1$ in terms of $\{P_{{\rm N},k}\}$ is given by \eqref{Eq:gradient}, as shown on the bottom of this page. The resultant problem is convex with respect to $\{P_{{\rm N},k}\}$, and thus can be efficiently solved by standard convex solvers such as CVX.
\begin{figure*}[b]
       		\hrulefill
\begin{equation}\label{Eq:gradient}
   \nabla_{P_k}B_1(\mathbf{p}) =- \frac{1}{\ln2}\sum_{j\neq k}\frac{N|h_{{\rm N},j}|^2\rho^2_{\rm NN}(\boldsymbol{\phi},\theta_k,r_k,\theta_j,r_j)}{\sum_{i\in\mathcal{K}\setminus\{j\}}P_{{\rm N},i}N|h_{{\rm N},j}|^2\rho^2_{\rm NN}(\boldsymbol{\phi},\theta_j,r_j,\theta_i,r_i)+\sum_{m\in\mathcal{M}}P_{{\rm F},m}N|h_{{\rm N},j}|^2\rho^2_{\rm NF}(\boldsymbol{\phi},\theta_j,r_j,\psi_m)+\sigma^2_{j}}.
\end{equation}
\end{figure*}
\vspace{-6pt}
\subsection{Outer Problem: Rotation Angle Optimization} 
With the power allocations optimized, the outer problem aims to maximize the sum-rate by optimizing the rotation angles $\boldsymbol{\phi}$. The outer problem is formulated as: 
\begin{subequations}
	\begin{align}
		({\bf P3}):~~~\max_{\substack{\boldsymbol{\phi}} }  ~~	\sum_{k \in \mathcal{K}}~R_{k}(\boldsymbol{\phi})
		~~~\text{s.t.}
       ~~~	\eqref{P1:rot_mutual},\eqref{P1:rot_region}.\nn
	\end{align}
 \end{subequations}
Solving problem (P3) directly  is difficult due to the lack of a closed-form expression for the power allocation and the intricate coupling of 
$\boldsymbol{\phi}$ in both near- and far-field steering vectors. Moreover, the search space for the rotation angles, defined as $[\phi_{q,{\rm min}},\phi_{q,{\rm max}}]^Q$, grows exponentially with the number of subarrays $Q$, leading to
prohibitively high computational complexity for an exhaustive search of
the optimal rotation angles.
To overcome these challenges, we adopt an efficient population-based optimization method, namely, PSO, to collectively explore the $Q$-dimensional solution space and determine the optimal rotation angles $\boldsymbol{\phi}$. The detailed procedures are outlined below.

Specifically, to implement the PSO-based approach, we initialize $S$ particles randomly, forming the initial swarm $\mathcal{S}^{(0)}=\{\boldsymbol{\phi}^{(0)}_1,\boldsymbol{\phi}^{(0)}_2,\cdots,\boldsymbol{\phi}^{(0)}_S\}$.
Each particle $s\in\mathcal{S}$ represents a feasible configuration of $Q$ rotation angles, expressed as:
\begin{equation}
    \boldsymbol{\phi}^{(0)}_s=[\phi^{(0)}_{s,1},\phi^{(0)}_{s,2},\cdots,\phi^{(0)}_{s,Q}]^T, 
\end{equation}
where $\phi^{(0)}_{s,q}$ denotes the rotation angle of the $q$-th subarray in the $s$-th particle, constrained to the admissible range $\phi^{(0)}_{s,q}\in [\phi_{q,{\rm min}},\phi_{q,{\rm max}}]$ for $1\le s \le S$ and $1\le q\le Q$. The initial velocity of each particle $s$ is defined as:
\begin{equation}
    \mathbf{v}^{(0)}_{s} = [v^{(0)}_{s,1},v^{(0)}_{s,2},\cdots,v^{(0)}_{s,Q}]^T.
\end{equation}
Let $\boldsymbol{\phi}_{s,{\rm p}}$ denote the individual best position of the $s$-th particle and $\boldsymbol{\phi}_{{\rm g}}$ the global best position across all particles. At each iteration $t$, the velocity and position of each particle are updated according to the PSO update rules:
\begin{align}
   {\mathbf{v}}^{(t+1)}_s &= \omega {\mathbf{v}}^{(t)}_s+c_1\tau_1(\boldsymbol{\phi}_{s,{\rm p}}-{\boldsymbol{\phi}}^{(t)}_s)+c_2\tau_2(\boldsymbol{\phi}_{{\rm g}}-{\boldsymbol{\phi}}^{(t)}_s), \\
   {\boldsymbol{\phi}}^{(t+1)}_s&={\boldsymbol{\phi}}^{(t)}_s+{\mathbf{v}}^{(t+1)}_s,
\end{align}
where $\omega$ is the inertia weight of the particle search.  Here,
$c_1$ and $c_2$ denote the individual and global learning factors, indicating the step sizes of each particle moving towards the individual and global best positions, respectively, while $\tau_1 \in [0,1]$ and $\tau_2 \in [0,1]$ are random values that introduce the search randomness. Moreover, to enforce constraint \eqref{P1:rot_mutual} of problem (P1), which requires that any two subarrays form an obtuse angle, rotation angles violating the admissible range are adjusted as follows:
\begin{align}
    [{\boldsymbol{\phi}}^{(t)}_s]_q=
    \begin{cases}
       \phi_{q,{\rm min}},  &\text{if} [{\boldsymbol{\phi}}^{(t)}_s]_q < \phi_{q,{\rm min}},\\ 
       \phi_{q,{\rm max}}, & \text{if} [{\boldsymbol{\phi}}^{(t)}_s]_q > \phi_{q,{\rm max}},\\
        [{\boldsymbol{\phi}}^{(t)}_s]_q,&\text{otherwise}.
    \end{cases}
\end{align}
The fitness of each particle is evaluated using the fitness function:
\begin{equation}
    f_{\rm fit}(\boldsymbol{\phi}^{(t)}_{s}) =  \sum_{k \in \mathcal{K}}~R_{k}(\boldsymbol{\phi}^{(t)}_{s})-\xi\mathbb{P}(\boldsymbol{\phi}^{(t)}_{s}),
\end{equation}
where $R_{k}(\boldsymbol{\phi})$ represents the sum-rate of near-field users, as defined in problem (P3). To penalize violations of constraint \eqref{P1:rot_mutual}, we introduce a penalty term to the fitness function, consisting of a penalty coefficient $\xi$ as well as the number of subarray pairs that violates the constraint \eqref{P1:rot_mutual}, which is given by\footnote{It is worth noting that, in practice, the admissible rotation range should be chosen appropriately by considering the mechanical constraints and the implementation costs. Specifically, given the sector-based nature of wireless coverage, users are typically located within a sector spanning from $\frac{\pi}{3}$ to $\frac{2\pi}{3}$ (i.e., tri-sector configurations) \cite{shao6d}. Accordingly, we in this paper adopt the same setting, with the rotation angles confined to $[-\frac{\pi}{6},\frac{\pi}{6}]$ in Section \ref{Sec:NR}. This choice naturally ensures that the constraint in \eqref{P1:rot_mutual} is satisfied.}
\begin{equation}
    \mathbb{P}(\boldsymbol{\phi}^{(t)}_{s}) = \sum_{p\in\mathcal{Q}}\sum_{q\in\mathcal{Q}\setminus{\{p\}}}\mathbb{I}\left(\left|[\boldsymbol{\phi}^{(t)}_{s}]_p-[\boldsymbol{\phi}^{(t)}_{s}]_q \right|<\frac{\pi}{2}\right),
\end{equation}
where $\mathbb{I}$ is an indicator function that  returns $1$
if the condition is true (i.e., the angle between subarrays $p$ and $q$ is less than $\pi/2$) and $0$ otherwise.
At each iteration, the fitness of each particle is computed, and the individual and global best positions are updated. This process continues until convergence is achieved.

\vspace{-9pt}

\subsection{Convergence and Computational Complexity}
   \subsubsection{Convergence}
   The proposed algorithm guarantees convergence, and the proof is provided as follows. At each iteration, the inner subproblem (P2) is solved optimally, and its objective value provides an upper bound for problem (P1). For the outer subproblem (P3), the fitness value of the global best particle in the evolving population $\mathcal{S}$, denoted $f_{\text{fit}}(\boldsymbol{\phi}^{(t)}{\text{g}})$, is non-decreasing across iterations, i.e., $f_{\text{fit}}(\boldsymbol{\phi}^{(t+1)}_{\text{g}}) \geq f_{\text{fit}}(\boldsymbol{\phi}^{(t)}_{\text{g}})$.
   Consequently, the objective value of problem (P1) is non-decreasing with each iteration. Given the finite total transmit power, convergence to a stationary point is ensured \cite{zhou2025rotatable}. 
   \subsubsection{Computational complexity}
   Solving the inner subproblem (P2) requires $I_2$ iterations and $K$ optimization variables, resulting in a complexity of $\mathcal{O}(I_2 K^{3.5})$ \cite{zhang2023swipt,dai2025rotatable}. For the outer subproblem, the PSO algorithm evolves a population of $S$ particles over $T$ iterations. Thus, the total computational complexity of the proposed double-layer algorithm for solving problem (P1) is $\mathcal{O}(S T I_2 K^{3.5})$.


\vspace{-9pt}
\section{Numerical Results}\label{Sec:NR}
In this section, we demonstrate the advantages of incorporating RAs into the mixed-field communication systems and the effectiveness of the proposed joint design. Unless otherwise stated, the system parameters are specified as: $N=129$, $f=24$ GHz, $K=3$, $M=2$, $S=100$, $T=100$, $L_{{\rm N},k} =L_{{\rm F},m} =3,\forall{k\in\mathcal{K},m\in\mathcal{M}}$, $\sigma^2_k=-70$ dBm, $\forall{k\in\mathcal{K}}$, and $C_{\phi_q} =[-\frac{\pi}{6},\frac{\pi}{6}],\forall{q\in\mathcal{Q}}$. The $K$ near-field users are uniformly distributed within a region defined by a distance range of $[0.05Z_{\rm Ray},0.2Z_{\rm Ray}]$ and a spatial angle range of $[\frac{\pi}{3},\frac{2\pi}{3}]$. The $M$ legacy far-field users are uniformly positioned within the same spatial angle range. The maximum BS transmit power is set to $P=30$ dBm, and the power allocated to each far-field user is $P_{{\rm F},m}=30$ dBm, $\forall{m}\in\mathcal{M}$.
For performance comparison, we consider the following benchmark schemes:

\begin{itemize}
   \item \emph{Fixed antenna with zero forcing (FA+ZF):} The rotation angle is fixed at zero, i.e., $\phi=0$, while transmit beamforming vectors are designed via the ZF technique.
    \item \emph{Fixed antenna with optimized power allocation (FA+OPA):} The rotation angle is fixed at zero, i.e., $\phi=0$, with power allocation optimized using the proposed approach.
    \item \emph{Rotatable antenna with ZF digital beamforming (RA+ZF):} The rotation angle is optimized using the proposed algorithm, while ZF beamforming is applied. 
    \item \emph{Rotatable antenna with equal power allocation (RA+EPA):} The rotation angle is optimized via the proposed algorithm, while power is equally distributed among near-field users. 
\end{itemize}

For the proposed joint design, we consider two cases corresponding to array rotation and subarray rotation, respectively, for performance comparison:  \emph{proposed joint design, $Q=5$} with $5$ RA subarrays and \emph{proposed joint design, $Q=1$} with $1$ RA subarray only.



\vspace{-9pt}
\subsection{Sum-rate Versus Maximum Transmit Power}
In Fig. \ref{Fig:BStrans}, we illustrate the effect of the transmit power on the sum-rate performance of various schemes. First, it is observed that the proposed scheme closely matches the performance of its combination with ZF digital beamforming and consistently outperforms the benchmark schemes, highlighting the effectiveness of RAs in enhancing the mixed-field communication performance. Second, compared to the RA+EPA scheme, which optimizes rotation alone, joint optimization of rotation and power allocation yields a higher sum-rate, underscoring the efficacy of the integrated approach. Third, subarray-wise rotation delivers significant performance gains over the array-wise rotation, owing to its superior interference mitigation capabilities enabled by greater design flexibility. Additionally, one interesting observation is that the FA+ZF scheme underperforms the proposed method across all transmit power levels. This is because, while FA+ZF effectively nullifies the near-field inter-user interference, it simultaneously suffers from the complex mixed-field interference caused by far-field users, resulting in notable performance degradation, especially when the mixed-field interference is pronounced.

\vspace{-9pt}
\subsection{Sum-rate Versus Power to Each Far-field User}
In Fig. \ref{Fig:Fartrans}, we examine the effect of power allocated to each far-field user, i.e., $P_{{\rm F},m}$, on the sum-rate. A clear trend emerges: all schemes exhibit a monotonic decrease in sum-rate as $P_{{\rm F},m}$ increases due to growing interference from far-field users. Notably, the proposed method with subarray rotation ($Q=5$) consistently outperforms the benchmark schemes under different values of $P_{{\rm F},m}$. 
In contrast, the proposed method with array rotation ($Q=1$) underperforms the FA+ZF scheme when $P_{{\rm F},m}\le23$ dBm but outperforms it when $P_{{\rm F},m}$ exceeds this threshold. This shift occurs because the FA+ZF approach effectively mitigates the near-field inter-user interference and exhibits superior performance when the mixed-field interference is minimal. However, as the growing mixed-field intensifies with higher far-field power, the performance of the FA+ZF scheme deteriorates significantly, highlighting the necessity of interference mitigation in such scenarios. Notably, we see that the largest performance gap between the proposed method (both array and subarray rotation) and benchmark schemes occurs at lower $P_{{\rm F},m}$ values, e.g., $P_{{\rm F},m}\le$ 20 dBm, demonstrating the efficacy of RAs in mitigating both near-field and mixed-field interference, especially in moderate interference conditions.
\begin{figure}[t]
	\centering
\includegraphics[width=0.45\textwidth]{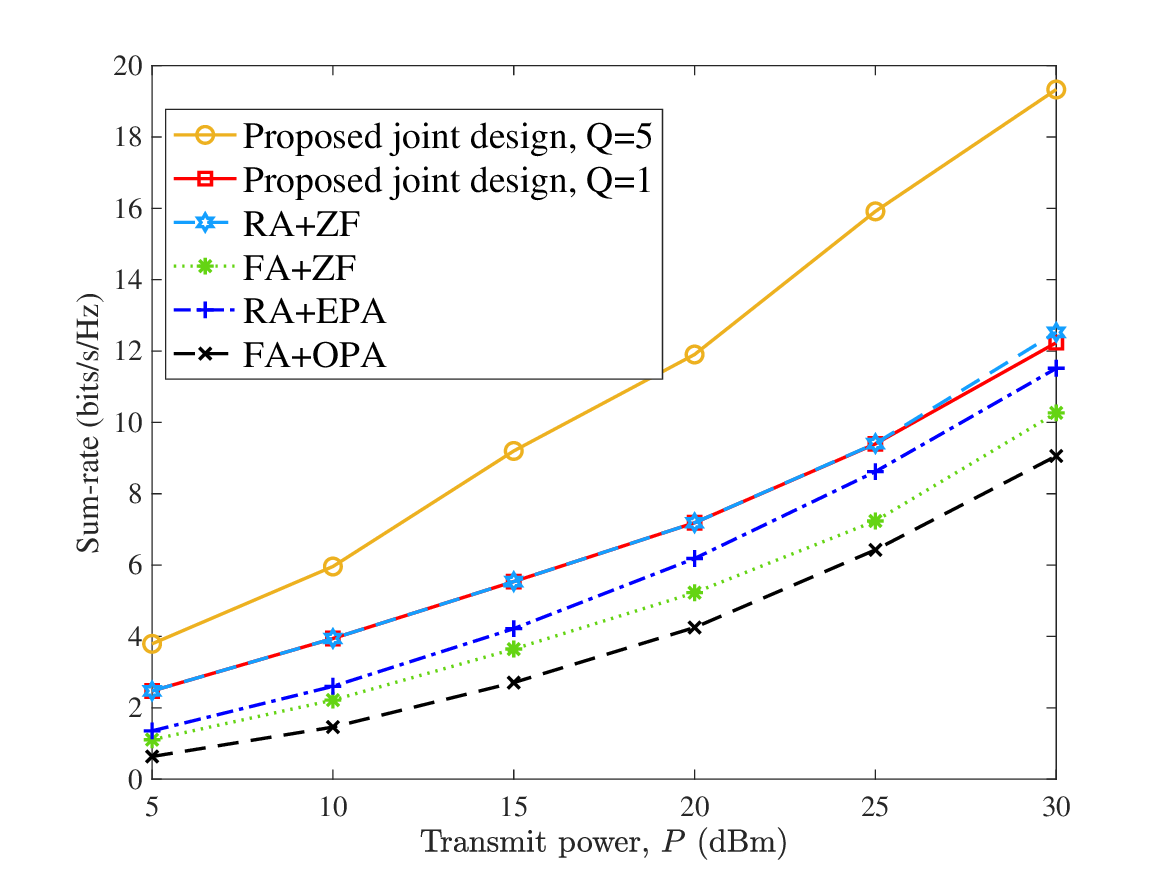}
	\caption{{Sum-rate versus transmit power.}}\label{Fig:BStrans}\vspace{-15pt}
\end{figure}
\begin{figure}[t]
	\centering
\includegraphics[width=0.45\textwidth]{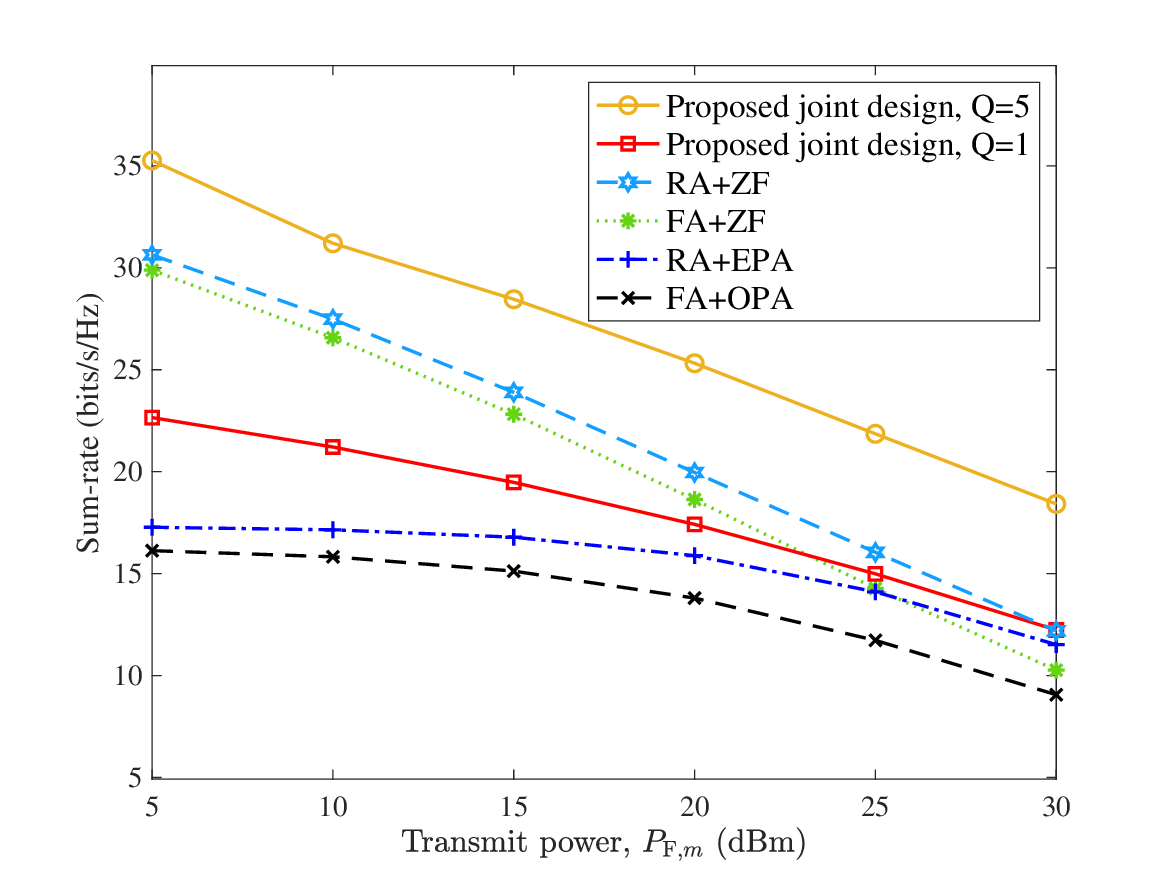}
	\caption{{Sum-rate versus power to each far-field user.}} \label{Fig:Fartrans}\vspace{-15pt}
\end{figure}
\begin{figure}[t]
	\centering
\includegraphics[width=0.45\textwidth]{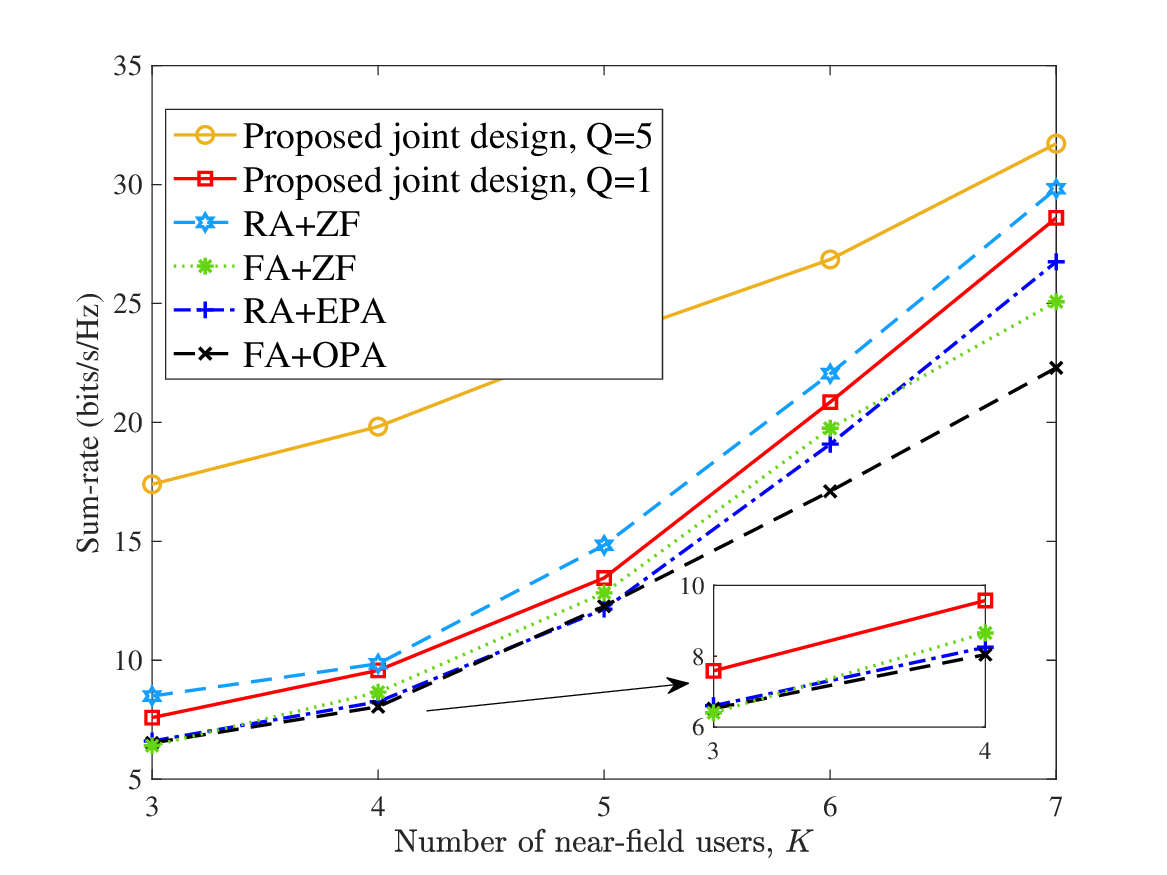}
	\caption{Sum-rate versus number of near-field users.} \label{Fig:NumNU}\vspace{-15pt}
\end{figure}
\begin{figure}[t]
	\centering
\includegraphics[width=0.45\textwidth]{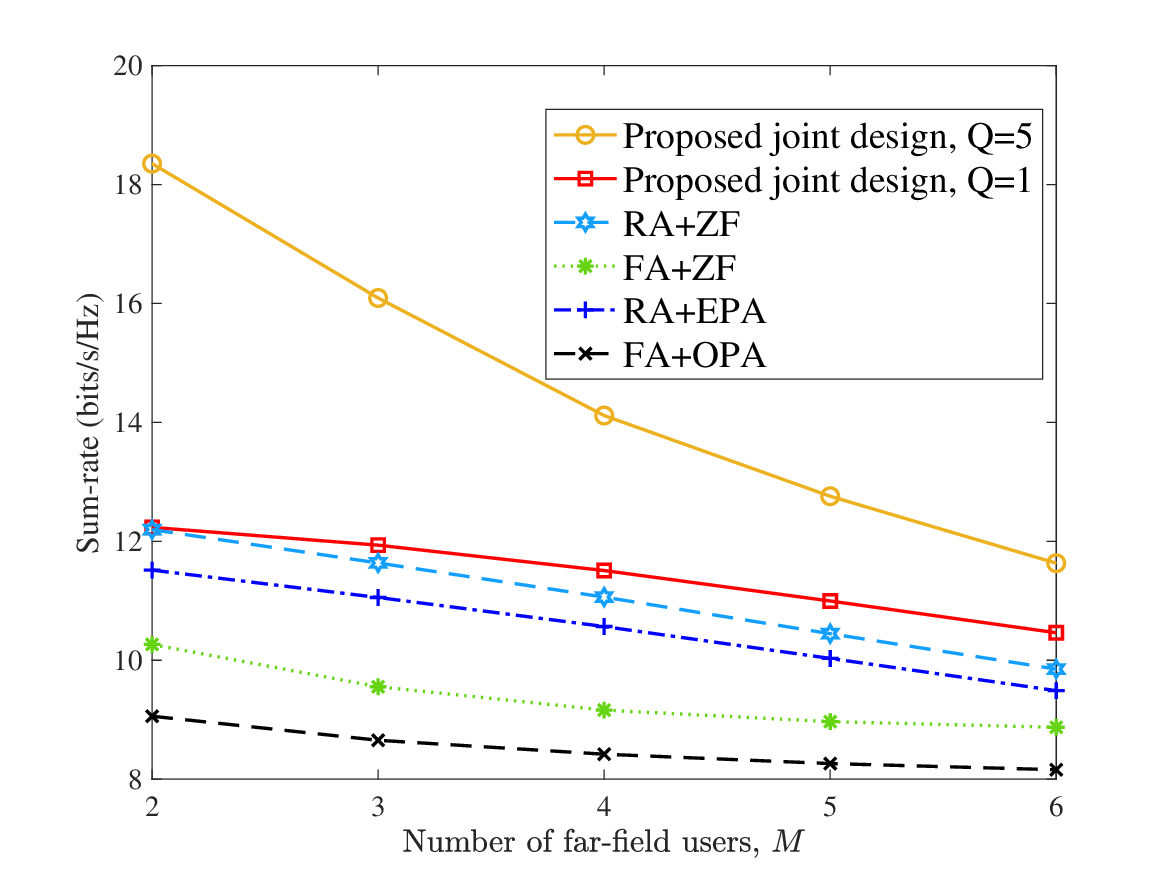}
	\caption{{Sum-rate versus number of far-field users.}}\label{Fig:NumFU}\vspace{-15pt}
\end{figure}
\vspace{-9pt}
\subsection{Sum-rate Versus Number of Near-field Users}
In Fig. \ref{Fig:NumNU}, we compare the sum-rate of different schemes versus the number of near-field users $K$. Beyond the initial three near-field users, additional near-field users are sequentially sampled from a region defined by a distance range of $[0.05Z_{\rm Ray},0.2Z_{\rm Ray}]$ and a spatial angle range of $[\frac{\pi}{3},\frac{2\pi}{3}]$. First, it is observed that the sum-rate of all schemes increases with the number of near-field users. This is because, despite the presence of the near-field and mixed-field interference, the system sum-rate is dominated by the achievable rates of near-field users, resulting in a higher sum-rate as $K$ increases. Second, our joint design with subarray rotation significantly outperforms both the proposed method with array rotation and the benchmark schemes. Specifically, when the number of near-field users is moderate, i.e., $K\le5$, the near-field interference remains manageable, and the enhanced design flexibility provided by subarray rotation yields significant performance gains over array rotation and other benchmarks. This highlights the superior design flexibility of the proposed scheme with $Q=5$ subarrays in mitigating near-field interference compared to a single subarray. However, the performance gap narrows as the number of near-field users further grows, i.e., $K>5$. This is because even with $Q=5$ subarrays, the advantage of the proposed method diminishes in scenarios with very high near-field user density, where the increasing interference complexity demands greater design flexibility for effective interference mitigation.

 \vspace{-9pt}
\subsection{Sum-rate Versus Number of Far-field Users}
In Fig. \ref{Fig:NumFU}, we plot the sum-rate of various schemes against the number of far-field users $M$. Beyond the initial two far-field users, additional far-field users are sampled from a spatial angle range of $[\frac{\pi}{3},\frac{2\pi}{3}]$. 
First, it is seen that the sum-rate of all schemes decreases monotonically with the number of far-field users. This is because the newly added far-field users progressively introduce severe mixed-field interference that degrades the performance of near-field users. Second, the proposed method with subarray rotation shows the most significant performance improvement in the small-$M$ regime. However, this advantage saturates as $M$ increases, e.g., $M\ge5$. The reason is that with a large number of far-field users, the mixed-field interference becomes increasingly complicated, thereby posing significant challenges to interference mitigation with respect to subarray rotation.

\vspace{-8pt}
 \section{Conclusion}\label{Sec:con}
In this paper, we proposed leveraging RAs to mitigate near-field and mixed-field interference by exploiting the spatial DoF enabled by RA array rotation. For the special case where all subarrays share a uniform rotation angle, we derived the closed-form expressions for the rotation-dependent near-field interference and the rotation-dependent mixed-field interference using the Fresnel integrals. In particular, we confirmed that array rotation effectively suppresses complex near-field and mixed-field interference, thus achieving substantial performance improvements in mixed-field communication scenarios. For the general case with subarray-wise rotation,  we have developed a double-layer algorithm, where the inner layer optimizes power allocation via SCA, and the outer layer determines subarray rotation angles via PSO. Numerical results demonstrated that the RA-enabled system substantially outperforms conventional fixed-antenna configurations and underscore the superior 
effectiveness of our proposed joint design compared to benchmark schemes. 
\vspace{-9pt}
\begin{appendices}
	\section{}\label{App1}
    First, we rewrite \eqref{Eq:NF_inter} as:
    \begin{align}\label{Eq:new}
         \rho_{\rm NN}(\phi,\theta_k,r_k,\theta_i,r_i) =\frac{1}{N}\left|\sum_{n=-\tilde{N}}^{\tilde{N}}e^{j\pi(A_1n+A_2)^2}\right|.
    \end{align}
    where $A_1 = \sqrt{d\left|\frac{\sin^2{(\phi-\theta_k)}}{2r_k}-\frac{\sin^2{(\phi-\theta_i)}}{2r_i}\right|}$ and $A_2= \frac{1}{2A_1}\left(\cos{(\phi-\theta_i)}- \cos{(\phi-\theta_k)}\right)$.

    The summation in \eqref{Eq:new}  can be approximated by the following integral:
    \begin{align}
         \rho_{\rm NN}&(\phi,\theta_k,r_k,\theta_i,r_i)  \stackrel{\left(a_1\right)}{\approx} \frac{1}{N}\left|\int_{-\tilde{N}}^{\tilde{N}} e^{j \pi\left(A_1 n+A_2\right)^2} \mathrm{d} n\right|,\nn\\
        &\stackrel{\left(a_2\right)}{=}\frac{1}{\sqrt{2}A_1N}\left|\int_{-\sqrt{2}A_1\tilde{N}+\sqrt{2}A_2}^{\sqrt{2}A_1\tilde{N}+\sqrt{2}A_2} e^{j \pi\frac{1}{2}t^2} \mathrm{d} t\right|,
    \end{align}
    where $(a_1)$ becomes increasingly accurate as $N\rightarrow\infty$ based on the Riemann integral, and $(a_2)$ is obtained  by substituting $(A_1n+A_2)^2=\frac{1}{2}t^2$.
   Then, utilizing the Fresnel integrals, we have 
   \begin{align}
        \!\! & \rho_{\rm NN}(\phi,\theta_k,r_k,\theta_i,r_i)
        =\frac{1}{\sqrt{2}A_1N}\left|\int_{-\sqrt{2}A_1\tilde{N}+\sqrt{2}A_2}^{\sqrt{2}A_1\tilde{N}+\sqrt{2}A_2} e^{j \pi\frac{1}{2}t^2} \mathrm{d} t\right|\nn\\
        \!\!&\!\!\stackrel{\left(a_3\right)}{=}\left|\frac{{C}(\beta_1+\beta_2)-{C}(\beta_1-\beta_2)+j({S}(\beta_1+\beta_2)-{S}(\beta_1-\beta_2))}{2\beta_2}\right|\nn\\
        \!\!&\!\!\stackrel{\left(a_4\right)}{=}\left|\frac{\widehat{C}(\beta_1,\beta_2)+j\widehat{S}(\beta_1,\beta_2)}{2\beta_2}\right|\triangleq G(\beta_1,\beta_2),
   \end{align}
    where $(a_3)$ is obtained by letting
    \begin{equation}
        \beta_1 =\frac{\left(\cos{(\phi-\theta_k)}-\cos{(\phi-\theta_i)}\right)}{\sqrt{d\left|\frac{\sin^2{(\phi-\theta_k)}}{r_k}-\frac{\sin^2{(\phi-\theta_i)}}{r_i}\right|}},
\end{equation}
and
    \begin{equation}
        \beta_2=\frac{N}{2} \sqrt{d\left|\frac{\sin^2{(\phi-\theta_k)}}{r_k}-\frac{\sin^2{(\phi-\theta_i)}}{r_i}\right|},
    \end{equation} 
    and $(a_4)$
     is obtained by defining  $\widehat{C}(\beta_1,\beta_2) = {C}(\beta_1+\beta_2)-{C}(\beta_1-\beta_2)$ and $\widehat{S}(\beta_1,\beta_2) = {S}(\beta_1+\beta_2)-{S}(\beta_1-\beta_2)$. 
   This completes the proof of Theorem \ref{The:NF_inter}.
   \vspace{-9pt}
    \section{}\label{App11}
    When $\theta_k=\theta_i$, the normalized near-field interference depends only on $\beta_2$, i.e., $ \rho_{\rm NN}(\phi,\theta,r_k,\theta,r_i) = G(0,\beta_2)$, where
\begin{equation}\label{Eq:NN_equal}
    \beta_2={N} \sqrt{d\left|\frac{\sin^2{(\phi-\theta)}}{r_k}-\frac{\sin^2{(\phi-\theta)}}{r_i}\right|}.
\end{equation} 
Based on the properties of the function $G(\cdot)$ in Remark \ref{Re:Properites}, the normalized near-field interference generally decreases as $\beta_2$ increases. From \eqref{Eq:NN_equal}, we see that the rotation angle influences the value of the sine function that attains its maximum value of one at $\frac{\pi}{2}$. Therefore, we have
\begin{itemize}
    \item When $\theta=\frac{\pi}{2}$, the optimal rotation angle $\phi$ is zero for achieving the maximum value of $\beta_2$.
    \item While for the case $\theta\neq\frac{\pi}{2}$, the optimal rotation angle is set to let $|\phi-\theta|$ as close to $\frac{\pi}{2}$ as possible within its admissible rotation range, i.e.,
            \begin{equation}
            \phi^* = \begin{cases}
            \min(\theta-\frac{\pi}{2},\phi_{\rm max}), & {\rm if}~ \theta>\frac{\pi}{2}, \\
                 \max(\theta-\frac{\pi}{2},\phi_{\rm min}), & {\rm if}~ \theta<\frac{\pi}{2}.
            \end{cases}
        \end{equation}
\end{itemize}
Combining the above leads to the desired result in Lemma \ref{Le:NN_equal}.

   \vspace{-9pt}
  \section{}\label{App2}
Without loss of generality, consider the case where   $\theta_i,\theta_k\in(0,\frac{\pi}{2})$. We begin with the case $\theta_i<\theta_k$, where $\phi\in(-\theta_i,\theta_i)$.
    Under the condition $r_i=r_k=r$ and using trigonometric identities, we express the four parameters as:
\begin{align}
\beta_1&=\frac{\left(\cos{(\phi-\theta_i)}-\cos{(\phi-\theta_k)}\right)}{\sqrt{\frac{d}{r}\left|\sin{(2\phi-\theta_k-\theta_i)}\right|\sin{(\theta_k-\theta_i)}}}, \nn\\
\beta_2&={N} \sqrt{\frac{d}{r}\left|\sin{(2\phi-\theta_k-\theta_i)}\right|\sin{(\theta_k-\theta_i)}},\nn\\
\bar{\beta}_1&=\frac{\left(\cos{\theta_i}-\cos{\theta_k}\right)}{\sqrt{\frac{d}{r}\sin{(\theta_k+\theta_i)}\sin{(\theta_k-\theta_i)}}},\nn\\
\bar{\beta}_2&={N} \sqrt{\frac{d}{r}\sin{(\theta_k+\theta_i)}\sin{(\theta_k-\theta_i)}}.
    \end{align}
To verify $\beta_1>\bar{\beta}_1$ and $\beta_2>\bar{\beta}_2$, we define:
    \begin{align}
        &\Delta_1(\phi) = \frac{\left(\cos{(\phi-\theta_i)}-\cos{(\phi-\theta_k)}\right)}{\sqrt{\left|\sin{(2\phi-\theta_k-\theta_i)}\right|}}-\frac{\left(\cos{\theta_i}-\cos{\theta_k}\right)}{\sqrt{|\sin{(\theta_k+\theta_i)}|}}, \nn\\
        &\Delta_2(\phi) = \left|\sin{(2\phi-\theta_k-\theta_i)}\right|-|\sin{(\theta_k+\theta_i)}|.
    \end{align}
It can be easily verified that $\beta_1>\bar{\beta}_1$ when $\Delta_1(\phi)>0$, and $\beta_2>\bar{\beta}_2$ when $\Delta_2(\phi)>0$. Since $\theta_i<\theta_k$, we rewrite:
    \begin{align}
        &\Delta_1(\phi) = \frac{\left(\cos{(\phi-\theta_i)}-\cos{(\phi-\theta_k)}\right)}{\sqrt{-\sin{(2\phi-\theta_k-\theta_i)}}}-\frac{\left(\cos{\theta_i}-\cos{\theta_k}\right)}{\sqrt{\sin{(\theta_k+\theta_i)}}}, \nn\\
        &\Delta_2(\phi) = -\sin{(2\phi-\theta_k-\theta_i)}-\sin{(\theta_k+\theta_i)}.
    \end{align}
    Let $\phi = -\epsilon$, where $0<\epsilon<\min\{\theta_i,\theta_k\}$. For $\Delta_2(\phi)$:
    \begin{align}
        \Delta_2(-\epsilon) = \sin{(2\epsilon+\theta_k+\theta_i)}-\sin{(\theta_k+\theta_i)}.
    \end{align}
    Since $\theta_k+\theta_i<\pi$,  for small $\epsilon>0$, we use the Taylor expansion to obtain: 
    \begin{equation}     \sin{(2\epsilon+\theta_k+\theta_i)}\approx \sin{(\theta_k+\theta_i)}+2\epsilon\cos{(\theta_k+\theta_i)}.
    \end{equation}
    Thus, 
    \begin{equation}
         \Delta_2(-\epsilon)\approx2\epsilon\cos{(\theta_k+\theta_i)}>0,
    \end{equation}
    since $\cos{(\theta_k+\theta_i)}>0$ for $\theta_k+\theta_i<\pi$. Hence, $\Delta_2(-\epsilon)>0$, implying $\beta_2>\bar{\beta}_2$.

    For $\Delta_1(\phi)$, we define:
    \begin{align}
        a(\epsilon) &= \cos{(\epsilon+\theta_i)}-\cos{(\epsilon+\theta_k)},\nn\\ 
        b(\epsilon) &= \sqrt{\sin{(2\epsilon+\theta_k+\theta_i)}}.
    \end{align}
   Employing the Taylor expansion for small $\epsilon$ yields:
    \begin{align}
        a(\epsilon) &\approx (\cos{\theta_i}-\cos{\theta_k}) +(\sin{\theta_k}-\sin{\theta_i})\epsilon,\nn\\
        b(\epsilon) &\approx \sqrt{\sin{(\theta_k+\theta_i)}}+\frac{\cos{(\theta_k+\theta_i)}}{\sqrt{\sin{(\theta_k+\theta_i)}}}\epsilon.
    \end{align}
    Since $\theta_i<\theta_k$, $\sin{\theta_k}-\sin{\theta_i}>0$, so $ a(\epsilon)>\cos{\theta_i}-\cos{\theta_k}$. We then compute:
    \begin{align}
        &\Delta_1(-\epsilon)  \nn\\
        &\approx \frac{\cos{\theta_i}-\cos{\theta_k}+(\sin{\theta_k}-\sin{\theta_i})\epsilon}{\sqrt{\sin{(\theta_k+\theta_i)}}+\frac{\cos{(\theta_k+\theta_i)}}{\sqrt{\sin{(\theta_k+\theta_i)}}}\epsilon}-\frac{\cos{\theta_i}-\cos{\theta_k}}{\sqrt{\sin{(\theta_k+\theta_i)}}}.
            \end{align}
            For small $\epsilon$, the dominant term is:
        \begin{equation}
            \Delta_1(-\epsilon)\approx \frac{(\sin{\theta_k}-\sin{\theta_i})\epsilon}{\sqrt{\sin{(\theta_k+\theta_i)}}}>0,
        \end{equation}
         since $\sin{\theta_k}-\sin{\theta_i}>0$. Thus, $\Delta_1(-\epsilon)>0$, implying $\beta_1>\bar{\beta}_1$.        
For $\theta_i>\theta_k$, the signs of  $\sin{(\theta_k-\theta_i)}$ and $\cos{\theta_i}-\cos{\theta_k}$ reverse, but the analysis holds analogously by adjusting the signs in $\Delta_1(\phi)$ and $\Delta_2(\phi)$. The case $\theta_i,\theta_k\in(\frac{\pi}{2},\pi)$ is similar, leveraging the symmetry of trigonometric functions. Thus, the inequalities hold for all cases, which completes the proof.
    
\end{appendices}
\vspace{-9pt}
\bibliographystyle{IEEEtran}
\bibliography{Ref}

\begin{thebibliography}{10}
\providecommand{\url}[1]{#1}
\csname url@samestyle\endcsname
\providecommand{\newblock}{\relax}
\providecommand{\bibinfo}[2]{#2}
\providecommand{\BIBentrySTDinterwordspacing}{\spaceskip=0pt\relax}
\providecommand{\BIBentryALTinterwordstretchfactor}{4}
\providecommand{\BIBentryALTinterwordspacing}{\spaceskip=\fontdimen2\font plus
\BIBentryALTinterwordstretchfactor\fontdimen3\font minus \fontdimen4\font\relax}
\providecommand{\BIBforeignlanguage}[2]{{%
\expandafter\ifx\csname l@#1\endcsname\relax
\typeout{** WARNING: IEEEtran.bst: No hyphenation pattern has been}%
\typeout{** loaded for the language `#1'. Using the pattern for}%
\typeout{** the default language instead.}%
\else
\language=\csname l@#1\endcsname
\fi
#2}}
\providecommand{\BIBdecl}{\relax}
\BIBdecl

\bibitem{10054381}
C.-X. Wang \emph{et~al.}, ``On the road to {6G}: Visions, requirements, key technologies, and testbeds,'' \emph{{IEEE} Commun. Surv. Tut.}, vol.~25, no.~2, pp. 905--974, 2nd Quart., 2023.

\bibitem{you2024next}
C.~You, Y.~Cai, Y.~Liu, M.~D.~Renzo, T.~M. Duman, A.~Yener, and A.~L.~Swindlehurst, ``Next generation advanced transceiver technologies for {6G} and beyond,'' \emph{{IEEE} J. Sel. Areas Commun.}, vol.~43, no.~3, pp. 582--627, Mar. 2025.

\bibitem{10559933}
J.~Zhang, H.~Miao, P.~Tang, L.~Tian, and G.~Liu, ``New mid-band for {6G}: Several considerations from the channel propagation characteristics perspective,'' \emph{{IEEE} Commun. Mag.}, vol.~63, no.~1, pp. 175--180, Jan. 2025.

\bibitem{tang2025preliminary}
P.~Tang, J.~Zhang, H.~Xu, H.~Miao, and X.~Liu, ``Preliminary perspectives on {3GPP} standardization of the propagation channel model for {FR3} bands for {NR},'' \emph{Sci. China. Inf. Sci.}, vol.~68, no.~3, Jun. 2025, Art. no. 137301.

\bibitem{10496996}
H.~Lu \emph{et~al.}, ``A tutorial on near-field {XL-MIMO} communications toward {6G},'' \emph{{IEEE} Commun. Surv. Tut.}, vol.~26, no.~4, pp. 2213--2257, 4th Quart., 2024.

\bibitem{you2023near}
C.~You, Y.~Zhang, C.~Wu, Y.~Zeng, B.~Zheng, L.~Chen, L.~Dai, and A.~L. Swindlehurst, ``Near-field beam management for extremely large-scale array communications,'' \emph{arXiv preprint arXiv:2306.16206}, 2023.

\bibitem{9903389}
M.~Cui, Z.~Wu, Y.~Lu, X.~Wei, and L.~Dai, ``Near-field communications for {6G}: Fundamentals, challenges, potentials, and future directions,'' \emph{IEEE Commun. Mag.}, vol.~61, no.~1, pp. 40--46, Jan. 2023.

\bibitem{liu2023near}
Y.~Liu, Z.~Wang, J.~Xu, C.~Ouyang, X.~Mu, and R.~Schober, ``Near-field communications: A tutorial review,'' \emph{IEEE Open J. Commun. Soc.}, vol.~4, pp. 1999--2049, Aug. 2023.

\bibitem{zhang2023mixed}
Y.~Zhang, C.~You, L.~Chen, and B.~Zheng, ``Mixed near- and far-field communications for extremely large-scale array: An interference perspective,'' \emph{{IEEE} Commun. Lett.}, vol.~27, no.~9, pp. 2496--2500, Sep. 2023.

\bibitem{10812003}
Y.~Xiao, E.~Wang, Y.~Chen, L.~Chen, A.~Ikhlef, and H.~Sun, ``Integrated sensing and communications with mixed fields using transmit beamforming,'' \emph{{IEEE} Wireless Commun. Lett.}, vol.~14, no.~3, pp. 726--730, Mar. 2025.

\bibitem{10129111}
Z.~Ding, R.~Schober, and H.~V. Poor, ``{NOMA}-based coexistence of near-field and far-field massive {MIMO} communications,'' \emph{{IEEE} Wireless Commun. Lett.}, vol.~12, no.~8, pp. 1429--1433, Aug. 2023.

\bibitem{zhang2023swipt}
Y.~Zhang and C.~You, ``{SWIPT} in mixed near- and far-field channels: Joint beam scheduling and power allocation,'' \emph{{IEEE} J. Sel. Areas Commun.}, vol.~42, no.~6, pp. 1583--1597, Jun. 2024.

\bibitem{11197972}
Y.~Zhang, C.~You, and H.~C. So, ``Movable-antenna position optimization: A new evolutionary framework,'' \emph{{IEEE} Trans. Wireless Commun.}, early access, 2025.

\bibitem{9264694}
K.-K. Wong, A.~Shojaeifard, K.-F. Tong, and Y.~Zhang, ``Fluid antenna systems,'' \emph{{IEEE} Trans. Wireless Commun.}, vol.~20, no.~3, pp. 1950--1962, Mar. 2021.

\bibitem{9650760}
K.-K. Wong and K.-F. Tong, ``Fluid antenna multiple access,'' \emph{{IEEE} Trans. Commun.}, vol.~21, no.~7, pp. 4801--4815, Jul. 2022.

\bibitem{10318061}
L.~Zhu, W.~Ma, and R.~Zhang, ``Modeling and performance analysis for movable antenna enabled wireless communications,'' \emph{{IEEE} Trans. Wireless Commun.}, vol.~23, no.~6, pp. 6234--6250, Jun. 2024.

\bibitem{10243545}
W.~Ma, L.~Zhu, and R.~Zhang, ``{MIMO} capacity characterization for movable antenna systems,'' \emph{{IEEE} Trans. Wireless Commun.}, vol.~23, no.~4, pp. 3392--3407, Apr. 2024.

\bibitem{10595399}
H.~Wang, Q.~Wu, and W.~Chen, ``Movable antenna enabled interference network: Joint antenna position and beamforming design,'' \emph{{IEEE} Commun. Lett.}, vol.~13, no.~9, pp. 2517--2521, Sep. 2024.

\bibitem{Zhu_Corre}
L.~Zhu, W.~Ma, and R.~Zhang, ``Movable-antenna array enhanced beamforming: Achieving full array gain with null steering,'' \emph{{IEEE} Commun. Lett.}, vol.~27, no.~12, pp. 3340--3344, Dec. 2023.

\bibitem{Ma_Corre}
W.~Ma, L.~Zhu, and R.~Zhang, ``Multi-beam forming with movable-antenna array,'' \emph{{IEEE} Commun. Lett.}, vol.~28, no.~3, pp. 697--701, Mar. 2024.

\bibitem{shao6d}
X.~Shao, Q.~Jiang, and R.~Zhang, ``{6D} movable antenna based on user distribution: Modeling and optimization,'' \emph{{IEEE} Trans. Wireless Commun.}, vol.~24, no.~1, pp. 355--370, Jan. 2025.

\bibitem{shao20246d}
X.~Shao, R.~Zhang, Q.~Jiang, and R.~Schober, ``{6D} movable antenna enhanced wireless network via discrete position and rotation optimization,'' \emph{{IEEE} J. Sel. Areas Commun.}, vol.~43, no.~3, pp. 674--687, Mar. 2025.

\bibitem{10883029}
X.~Shao, R.~Zhang, Q.~Jiang, J.~Park, T.~Q.~S. Quek, and R.~Schober, ``Distributed channel estimation and optimization for {6D} movable antenna: Unveiling directional sparsity,'' \emph{{IEEE} J. Sel. Topics Signal Process.}, vol.~19, no.~2, pp. 349--365, Mar. 2025.

\bibitem{10945745}
X.~Shao and R.~Zhang, ``{6DMA} enhanced wireless network with flexible antenna position and rotation: Opportunities and challenges,'' \emph{{IEEE} Commun. Mag.}, vol.~63, no.~4, pp. 121--128, Apr. 2025.

\bibitem{zheng2025rotatableM}
B.~Zheng, T.~Ma, C.~You, J.~Tang, R.~Schober, and R.~Zhang, ``Rotatable antenna enabled wireless communication and sensing: Opportunities and challenges,'' \emph{{IEEE} Wireless Commun.}, early access, 2025.

\bibitem{zheng2025rotatable}
B.~Zheng, Q.~Wu, and R.~Zhang, ``Rotatable antenna enabled wireless communication: Modeling and optimization,'' \emph{arXiv preprint arXiv:2501.02595}, 2025.

\bibitem{xie2025thz}
Y.~Xie, W.~Mei, D.~Wang, B.~Ning, Z.~Chen, J.~Fang, and W.~Guo, ``{THz} beam squint mitigation via {3D} rotatable antennas,'' \emph{arXiv preprint arXiv:2503.08134}, 2025.

\bibitem{dai2025rotatable}
L.~Dai, B.~Zheng, Q.~Wu, C.~You, R.~Schober, and R.~Zhang, ``Rotatable antenna-enabled secure wireless communication,'' \emph{{IEEE} Wireless Commun. Lett.}, early access, 2025.

\bibitem{zhou2025rotatable}
C.~Zhou, C.~You, B.~Zheng, X.~Shao, and R.~Zhang, ``Rotatable antennas for integrated sensing and communications,'' \emph{{IEEE} Wireless Commun. Lett.}, vol.~14, no.~9, pp. 2838--2842, Sep. 2025.

\bibitem{10960698}
K.~Qu, H.~Li, C.~Sun, W.~Zhang, S.~Guo, and H.~Zhang, ``Rotatable array-enabled multi-{BS} cooperative {ISAC} transmit beampattern design,'' \emph{{IEEE} Trans. Veh. Technol.}, vol.~74, no.~9, pp. 14\,775--14\,780, Sep. 2025.

\bibitem{zhang2022fast}
Y.~Zhang, X.~Wu, and C.~You, ``Fast near-field beam training for extremely large-scale array,'' \emph{IEEE Wireless Commun. Lett.}, vol.~11, no.~12, pp. 2625--2629, Dec. 2022.

\bibitem{9598863}
X.~Wei and L.~Dai, ``Channel estimation for extremely large-scale massive {MIMO}: Far-field, near-field, or hybrid-field?'' \emph{{IEEE} Commun. Lett.}, vol.~26, no.~1, pp. 177--181, Jan. 2022.

\bibitem{dai_LDMA}
Z.~Wu and L.~Dai, ``Multiple access for near-field communications: {SDMA} or {LDMA}?'' \emph{{IEEE} J. Sel. Areas Commun.}, vol.~41, no.~6, pp. 1918--1935, Jun. 2023.

\bibitem{polk1956optical}
C.~Polk, ``Optical {F}resnel-zone gain of a rectangular aperture,'' \emph{IRE Trans. Antennas Propag.}, vol.~4, no.~1, pp. 65--69, Jan. 1956.

\end{thebibliography}
\end{document}